\documentclass[9pt]{sigplanconf}

\usepackage[utf8]{inputenc}
\usepackage[T1]{fontenc}
\usepackage{microtype}
\usepackage{bm}
\usepackage{pgfplots}
\usepackage{enumerate,paralist}
\usepackage{pbox}
\usepackage{stmaryrd}

\usepackage{appendix}
\usepackage{comment}
\usepackage{parskip}
\usepackage{multirow}
\usepackage{url}

\usepackage[ruled, linesnumbered]{algorithm2e}
\usepackage{tikz}

\usepackage{calc}
\usepackage{fp}

\usepackage{lmodern}
\usepackage{amsmath}
\usepackage{amssymb, amsthm}
\usepackage{listings}
\usepackage{mathtools}
\usepackage{times}
\usepackage{multibib}
\usepackage{flushend}

\theoremstyle{plain}
\newtheorem{lemma}{Lemma}

\newtheorem{definition}{Definition}
\newtheorem{corollary}{Corollary}
\newtheorem{theorem}{Theorem}
\newtheorem{example}{Example}
\newtheorem{remark}{Remark}

\newcommand{\expv}{\mathbb{E}}

\newcommand{\program}{\mathcal{P}}

\lstdefinelanguage{affprob}
{
morekeywords={angel,demon, choice,prob(0.6), if, then, else, fi, while, do, od, 
true, false, and, or, skip, sample},
sensitive = false
}

\usepackage{tikz}
\usepackage{tikz,pgffor}
\usetikzlibrary{arrows}
\usetikzlibrary{shapes}
\usetikzlibrary{calc}
\usetikzlibrary{automata}
\usetikzlibrary{positioning}

\tikzstyle{ang}=[regular polygon, regular polygon sides = 3,draw,inner sep=0pt,minimum size=6mm, yshift = -0.75 mm]
\tikzstyle{dem}=[shape=diamond,draw,inner sep=0pt,minimum size=6mm]
\tikzstyle{ran}=[shape=circle,draw,inner sep=0pt,minimum size=5mm]
\tikzstyle{det}=[shape=rectangle,draw,inner sep=0pt,minimum size=5mm]
\tikzstyle{tran}=[draw,->,>=stealth, rounded corners]

\usepackage{todonotes}

\newcommand{\defineNote}[3][black!65!green]{\expandafter\def\csname #2\endcsname
##1{\stepcounter{fixcount}\fxwarning{\textcolor{#1}{\textbf{#3}: ##1}}}}

\DeclareMathAlphabet{\mathpzc}{OT1}{pzc}{m}{it}

\newcommand{\theoremlike}[2]{\par\medskip\penalty-250%
{{\bfseries\noindent
#2 \ref{#1}.}}\it}

\newcommand{\thmhelperpre}[2]{\theoremlike{#1}{#2}}
\newcommand{\thmhelperpost}{\par\medskip}

\newenvironment{reftheorem}[1]{\thmhelperpre{#1}{Theorem}}{\thmhelperpost}
\newenvironment{reflemma}[1]{\thmhelperpre{#1}{Lemma}}{\thmhelperpost}

\renewcommand{\vec}[1]{\mathbf{#1}}

\usepackage{graphicx}

\newcommand{\E}{\ensuremath{{\rm \mathbb E}}}

\renewcommand{\phi}{\varphi}
\newcommand{\exsup}{\triangleleft}

\newcommand{\eps}{\epsilon}

\newcommand{\pvars}{V}

\newcommand{\locs}{\mathit{L}}

\newcommand{\loc}{\ell}

\renewcommand{\E}{\mathbb{E}}

\newcommand{\Rset}{\mathbb{R}}
\newcommand{\Nset}{\mathbb{N}}
\newcommand{\Zset}{\mathbb{Z}}

\newcommand{\lin}{\mathit{in}}
\newcommand{\lout}{\mathit{out}}

\newcommand{\transitions}{\mapsto}
\newcommand{\probdist}{\mathit{Pr}}
\newcommand{\guards}{G}

\newcommand{\prob}{\mathit{Pr}}

\newcommand{\id}{\mathit{id}}
\newcommand{\probm}{\mathbb{P}}
\newcommand{\inv}{I}
\newcommand{\pinv}{PI}

\newcommand{\fail}[2]{\mathit{Fail}(#1)}

\newcommand{\APP}{{\sc App}}

\newcommand{\support}{\mathit{supp}}
\newcommand{\val}{\nu}

\newcommand{\vars}{\mathcal{V}}
\newcommand{\sat}[1]{\llbracket #1 \rrbracket}
\newcommand{\exsat}[1]{\sat{#1}^{\exsup}}
\newcommand{\up}{u}
\newcommand{\pCFG}{\mathcal{C}}

\newcommand{\sigmat}{\sigma_t}
\newcommand{\sigmaa}{\sigma_a}
\newcommand{\cfg}[2]{\vec{C}^{#1}_{#2}}
\newcommand{\lr}[2]{\loc^{#1}_{#2}}
\newcommand{\vr}[2]{\vec{x}^{#1}_{#2}}
\newcommand{\run}{\varrho}
\newcommand{\locinit}{\loc_{\mathit{init}}}
\newcommand{\vecinit}{\vec{x}_{\mathit{init}}}
\newcommand{\natfilt}{\mathcal{R}}
\newcommand{\lem}{\eta}
\newcommand{\preexp}[1]{\mathit{pre}_{#1}}
\newcommand{\Out}{\mathit{Out}}
\newcommand{\confset}{C}
\newcommand{\stime}{T}
\newcommand{\treach}[1]{\stime^{#1}}
\newcommand{\ttime}{\mathit{Term}}

\newcommand{\setmap}{\mathit{IC}}

\newcites{add}{Additional References}
\begin{document}
\toappear{} 

\title{Stochastic Invariants for Probabilistic Termination}
\authorinfo{Krishnendu Chatterjee}{IST Austria, 
	Klosterneuburg,
	Austria}{Krishnendu.Chatterjee@ist.ac.at}

\authorinfo{Petr Novotn\'{y}}{IST Austria, Klosterneuburg, 
Austria}{petr.novotny@ist.ac.at}

\authorinfo{\DJ{}or\dj{}e \v{Z}ikeli\'c}{University of 
	Cambridge, UK}{dz277@cam.ac.uk}

\maketitle




\begin{abstract}
Termination is one of the basic liveness properties, and we study
the termination problem for probabilistic programs with real-valued 
variables.
Previous works focused on the qualitative problem that asks 
whether an input program terminates with probability~1 (almost-sure termination).
A powerful approach for this qualitative problem is the notion of ranking 
supermartingales with respect to a given set of invariants.
The quantitative problem (probabilistic termination) asks for bounds on 
the termination probability, and this problem has not been addressed yet. 
A fundamental and conceptual drawback of the existing approaches to address 
probabilistic termination is that even though the supermartingales consider the 
probabilistic behaviour of the programs, the invariants are obtained completely 
ignoring the probabilistic aspect (i.e., the invariants are obtained 
considering all behaviours with no information about the probability).

In this work we address the probabilistic termination problem 
for linear-arithmetic probabilistic programs with nondeterminism.
We formally define the notion of {\em stochastic invariants}, which are 
constraints along with a probability bound that the constraints hold. 
We introduce a concept of {\em repulsing supermartingales}.
First, we show that repulsing supermartingales can be used to obtain bounds 
on the probability of the stochastic invariants. 
Second, we show the effectiveness of repulsing supermartingales in the following 
three ways:
(1)~With a combination of ranking and repulsing supermartingales we can compute
lower bounds on the probability of termination;
(2)~repulsing supermartingales provide witnesses for refutation of almost-sure 
termination;
and (3)~with a combination of ranking and repulsing supermartingales we can 
establish persistence properties of probabilistic programs.

Along with our conceptual contributions, we establish the following computational results:
First, the synthesis of a stochastic invariant which supports some ranking 
supermartingale and at the same time admits a repulsing supermartingale can 
be achieved via reduction to the existential first-order theory of reals, 
which generalizes existing results from the non-probabilistic setting. 
Second, given a program with ``strict invariants'' (e.g., obtained via abstract 
interpretation) and a stochastic invariant, we can check in polynomial time 
whether there exists a linear repulsing supermartingale w.r.t. the stochastic 
invariant (via reduction to LP).
We also present experimental evaluation of our approach on academic examples.
\end{abstract}

\section{Introduction}\label{sec:introduction}

\noindent{\em Probabilistic programs.}
There is a huge recent interest in the formal analysis of probabilistic programs, 
since they provide a rich framework to model a wide variety of applications ranging 
from randomized algorithms~\cite{RandBook,RandBook2}, to stochastic network 
protocols~\cite{BaierBook,prism}, robot 
planning~\cite{KGFP09,kaelbling1998planning},
or modelling problems in machine learning~\cite{GHNR14:prob-programming}, to 
name 
a few.
The extension of the classical imperative programs with \emph{random value generators}, 
that produce random values according to some desired probability distribution, gives
rise to probabilistic programs. 
The formal analysis of such programs, and probabilistic systems in general, gives 
rise to a wealth of research questions, which have been studied across diverse fields, 
such as probability theory and statistics~\cite{Durrett,Howard,Kemeny,Rabin63,PazBook}, 
formal methods~\cite{BaierBook,prism}, 
artificial intelligence~\cite{LearningSurvey,kaelbling1998planning}, 
and programming 
languages~\cite{SriramCAV,HolgerPOPL,SumitPLDI,EGK12,CFNH16:prob-termination}.

\smallskip\noindent{\em Termination problem.}
{\em Termination} is one of the most basic liveness properties for programs.
For non-probabilistic programs the proof for termination coincides with 
the construction of a \emph{ranking function}~\cite{rwfloyd1967programs}, 
and many different approaches exist for construction of ranking functions 
for non-probabilistic programs~\cite{DBLP:conf/cav/BradleyMS05,DBLP:conf/tacas/ColonS01,DBLP:conf/vmcai/PodelskiR04,DBLP:conf/pods/SohnG91}.
For probabilistic programs there are many natural extensions of the termination 
problem. 
The two most natural questions related to the probability of termination of an
input program are the \emph{qualitative} and \emph{quantitative} problems which are as 
follows: 
\begin{compactenum}
\item \emph{Qualitative problem: almost-sure termination.} 
The basic qualitative question is the \emph{almost-sure termination}
problem that asks whether the program terminates with probability~1~\cite{HolgerPOPL,BG05}.

\item \emph{Quantitative problem: probabilistic termination.} 
The natural generalization of the qualitative question is the quantitative
question of \emph{probabilistic termination} that asks for a lower bound 
on the probability of termination of the program.
\end{compactenum}
The above questions are the basic and fundamental questions for the static 
analysis of probabilistic programs.

\smallskip\noindent{\em Nondeterminism in probabilistic programs.}
The role of {\em nondeterminism} is also quite fundamental in probabilistic programs.
The nondeterminism is necessary in many cases such as for abstraction.
For efficient static analysis of large programs, it is infeasible to track all the variables.
Abstraction allows to ignore some variables, and for the sake of analysis  the
worst-case behaviour 
must be considered for them, which is modelled as nondeterminism.
Besides the modelling aspects, the presence of nondeterminism significantly 
changes 
the landscape of theoretical results, which we discuss below.

\smallskip\noindent{\em Previous results: almost-sure termination.} 
Given the importance of the termination problem for probabilistic programs,
the problem has been studied in great depth.
However, much of the previous research focused on the 
qualitative problem.
The details are as follows:
\begin{compactitem}
\item \emph{Discrete probabilistic choices.}
First~\cite{MM04,MM05} presented techniques for termination of probabilistic 
programs with nondeterminism, but restricted only to discrete probabilistic choices.
 
\item \emph{Infinite probabilistic choices without nondeterminism.}
The approach of~\cite{MM04,MM05} was extended in~\cite{SriramCAV} to \emph{ranking} 
martingales and supermartingales. 
The approach of~\cite{SriramCAV} presents a sound (but not complete) approach for 
almost-sure termination of infinite-state probabilistic programs (without nondeterminism) 
with integer and real-valued random variables drawn from distributions including 
uniform, Gaussian, and Poison.
The connection of termination of probabilistic programs without nondeterminism 
to \emph{Lyapunov ranking functions} was established in~\cite{BG05}.
For probabilistic programs with countable state space and without 
nondeterminism, the Lyapunov ranking functions provide a sound and complete method 
to prove termination in finite time, which implies almost-sure termination~\cite{BG05,Foster53}.
Another sound approach~\cite{DBLP:conf/sas/Monniaux01} for almost-sure termination 
is to explore the exponential decrease of probabilities upon bounded-termination 
through abstract interpretation~\cite{DBLP:conf/popl/CousotC77}.

\item \emph{Infinite probabilistic choices with nondeterminism.}
For probabilistic programs with nondeterminism the theoretical results change
significantly.
The Lyapunov ranking function method as well as the ranking martingale method 
are sound but not complete in the presence of nondeterminism~\cite{HolgerPOPL}.
Finally, for probabilistic programs with nondeterminism, a sound
and complete (for a well-defined class of probabilistic programs) 
characterization for almost-sure termination is obtained 
in~\cite{HolgerPOPL}, by generalizing the ranking supermartingale 
approach of~\cite{SriramCAV}.
The question of algorithmic synthesis of ranking supermartingales has also been 
considered, for probabilistic programs with linear arithmetic, and special 
classes of ranking supermartingales (such as linear and polynomial ranking 
supermartingales~\cite{CFNH16:prob-termination,CFG16:positivstellensatz-arxiv}).

\end{compactitem}
In all the existing approaches above for infinite-state probabilistic programs 
with non-determinism, 
the key technique for almost-sure termination is the notion of a ranking 
supermartingale (RSM). Intuitively, a ranking supermartingale is a function 
assigning numbers to program configurations (where each configuration consists 
of the current 
control location and current 
valuation of program variables) with the 
following property: in each reachable configuration, the expected value of the 
RSM in 
the next execution step is strictly smaller than its current value. Thus, RSMs 
form a probabilistic counterpart of classical ranking functions.

\smallskip\noindent{\em RSMs with respect to invariants.}
Since precisely characterizing the set of reachable configurations is 
infeasible in 
practice,
the previous works for almost-sure termination of infinite-state probabilistic 
programs 
consider the existence of ranking supermartingales (RSMs) with respect to invariants. 
An {\em invariant} is a set of constraints on the variables of the program, one 
for 
each program 
location, such that along {\em all} executions of the program, if a program 
location is 
visited, then the program variables must satisfy the constraints of the respective 
program location. Hence, each invariant represents an over-approximation of the 
set of reachable configurations.
The computational problem for almost-sure termination is to decide the existence of
a RSM for a probabilistic program w.r.t. an input invariant, i.e. a function 
assigning numbers to configurations such that for each configuration in the 
invariant, the 
expected value of the RSM in the next step is smaller than the current one.

\smallskip\noindent{\em RSMs and probabilistic termination.} In the 
\emph{probabilistic termination} problem we are interested in computing 
termination
probabilities when the program does not terminate almost-surely.
While reasoning about termination probabilities of probabilistic programs was 
considered before (at least on a theoretical level, see also Related Work 
section), approaches based on RSMs were not yet considered for this purpose.
A fundamental and conceptual problem here is that while RSMs take into account 
the probabilistic behaviour of the program, the invariants completely 
ignore the probabilistic aspect as they must hold along all executions
(i.e., the invariants are obtained considering all behaviours without any 
information about the probability).
Since all previous works on RSMs consider RSMs w.r.t. invariants, this implies 
a fundamental limitation of this tool to address probabilistic 
termination. 
We illustrate this with an example below.

\lstset{language=affprob}
\lstset{tabsize=3}
\newsavebox{\nonxxterm}
\begin{lrbox}{\nonxxterm}
\begin{lstlisting}[mathescape]
$x:=10$
while $x\geq 0$ do
	if $x \leq 100$ then $x:=x+$sample$(\mathrm{Uniform}[-2,1])$
	else $x:=x+$sample$(\mathrm{Uniform}[-1,2])$
    fi
od
\end{lstlisting}
\end{lrbox}
\begin{figure}[t]
\usebox{\nonxxterm}
\begin{tikzpicture}[x = 1.7cm]
\end{tikzpicture}
\caption{A probabilistic program modeling a generalization of an asymmetric 
one-dimensional random 
walk.}\label{fig:intro}
\end{figure}

\smallskip\noindent{\em Motivating example.}
Consider the probabilistic program shown in Figure~\ref{fig:intro}, which is an asymmetric one-dimensional
random walk. 
The random walk is denoted by value $x$. 
If $x$ is smaller than~100, then its value is incremented by a number uniformly chosen between $[-2,1]$,
otherwise the increment is uniform in $[-1,2]$.
In this random walk, $x$ can have any value above~0. But once the value reaches~100, with high probability
the value drifts away, and the program does not terminate.
In this example, there is no effective invariant, as $x$ can have any value. 
However, the assertion $x\leq 100$ is violated only with very small 
probability, and as long as $x\leq 100$ holds, the value of $x$ tends to 
decrease \emph{on average}.

\smallskip\noindent{\em Our contributions.} 
In this work we consider the probabilistic termination problem for 
linear-arithmetic 
probabilistic programs  with nondeterminism.
Our contributions are manifold, ranging from (a)~definition of stochastic invariants
for probabilistic termination; to (b)~introduction of {\em repulsing supermartingales}
(RepSMs) and their effectiveness; to (c)~computational results; and (d)~experimental results.
We describe each of them in details below.

\smallskip\noindent{\em Stochastic invariants.}
We formally define the notion of {\em stochastic invariants} for the probabilistic 
termination problem. 
A stochastic invariant consists of a constraint on the program variables for
each program location (as for invariants), and a threshold value $p$,
such that the constraint is violated at the location with probability at most 
$p$. 
For example, in the probabilistic program of Figure~\ref{fig:intro} 
we can consider a stochastic invariant with constraint $x \leq 100$ at location 
corresponding to the \textbf{if} $x\leq100$ test, with the threshold value 
being very  
small (less than $10^{-5}$),
since the probability that $x$ exceeds~100 is very low due to asymmetry. 

\smallskip\noindent{\em Repulsing supermartingales.} 
We introduce a concept of repulsing supermartingales (RepSMs), 
which are in some sense dual to RSMs.
A RepSM for a set of program configurations $C$ has non-negative value inside 
$C$ and decreases on average outside $C$. Intuitively, while RSMs show that a 
program execution cannot avoid some set $C$ of configurations indefinitely, 
RepSMs show that program executions that start outside of $C$ tend to avoid 
$C$, and that they actually tend to ``run away'' from $C$ in some well defined 
sense. The RepSMs are inspired by martingale methods used for analysing 
so-called one-counter MDPs \cite{BKK:oc-jacm,BBEK:oc-approx-IC}, 
but they are more general and apply to 
vastly larger class of systems.
Our results for RepSMs are as follows:
\begin{compactenum}
\item {\em Stochastic invariants.} We show that RepSMs can be used to 
obtain bounds on the probability threshold of the stochastic invariants.

\item {\em Effectiveness.} We show the effectiveness of RepSMs 
in the following three ways:
\begin{compactitem}
\item First, with a combination of RSMs and RepSMs we show how to obtain lower bounds on 
the probability of termination (i.e., sound bounds for probabilistic termination).
Hence %
for programs that do not terminate almost-surely, but with high probability, 
our method can obtain such bounds.

\item Second, in program analysis, refuting a property is as important as proving,
as refutation is important in bug-hunting. 
We show that RepSMs can provide witnesses for refuting almost-sure termination.
Moreover, even for programs that terminate almost-surely, but have infinite
expected termination time, RepSMs can serve as witnesses for infinite expected
termination time.
\item Finally, we show the effectiveness of RepSMs beyond the termination problem.
For reactive systems that are non-terminating a very basic property is 
persistence, which requires that the execution eventually stays in a desired 
set of configurations. 
We show that a combination of RSMs and RepSMs can establish persistence properties
of probabilistic programs.
\end{compactitem}
\end{compactenum}

\smallskip\noindent{\em Computational results.}
We present two computational results.
\begin{compactenum}
\item {\em Repulsing supermartingales w.r.t. stochastic invariants.}
First, we consider the problem of efficient algorithms for deciding the 
existence of RepSMs w.r.t. to stochastic invariants. 
Since our goal is to obtain efficient algorithms, we consider the simplest 
class of
RepSMs, namely, {\em linear repulsing supermartingales} (LRepSMs).
We show that given a program with ``strict" invariants'' (e.g., obtained via 
abstract 
interpretation) and a stochastic invariant, the existence of a LRepSM 
w.r.t. the stochastic invariant can be decided in polynomial time (via 
reduction 
to LP) provided that the stochastic invariant uses only \emph{polyhedral} 
constraints (i.e. conjunctions of inequalities).
\item {\em Synthesis.}
Second, we consider the problem of synthesis of a stochastic invariant which 
supports some 
RSM and at the same time admits a RepSM.
We show that the synthesis problem can be achieved via reduction to the existential first-order 
theory of reals.
This result generalizes existing results from the non-probabilistic setting, and  
even in the non-probabilistic setting the best-known computational methods require the 
existential theory of reals.

\end{compactenum}

\smallskip\noindent{\em Experimental results.} We present a basic implementation 
of our approach, and present experimental results on academic examples.
Our main contributions are conceptual and algorithmic, and the experiments 
serve as a validation of the new concepts.

Due to space constraints, some technical details are presented in the appendix.

\section{Preliminaries}\label{sec:prelim}

\subsection{Basic Notions, Linear Predicates, Valuations}
For a set $A$ we denote by $|A|$ the cardinality of $A$. We denote by $\Nset$,
$\Nset_0$, $\Zset$, and $\Rset$ the sets of all positive integers, non-negative
integers, integers, and real numbers, respectively. We assume basic knowledge 
of matrix calculus.
We use boldface notation for
vectors, e.g. $\vec{x}$, $\vec{y}$, etc., and we denote an $i$-th component of a
vector $\vec{x}$ by $\vec{x}[i]$. For the purpose of matrix calculations we 
assume that (non-transposed) vectors are row vectors. If $\vec{v},\vec{v}'$ are 
$n$ and $m$ dimensional vectors, respectively, then $(\vec{v},\vec{v}')$ is an 
$(n+m)$-dimensional vector obtained by ``concatenation'' of $\vec{v}$ and 
$\vec{v}'$. We identify 1-dimensional vectors with numbers. For an 
$n$-dimensional vector 
$\vec{x}$, index $1 \leq i\leq n$, and number $a$ we denote by $\vec{x}(i\leftarrow a)$ 
a 
vector $\vec{y}$ such that $\vec{y}[i]=a$ and $\vec{y}[j]=\vec{x}[j]$ for all 
$1\leq j \leq n$, $j\neq i$.
For comparison of vectors (e.g. as in $\vec{x}\leq \vec{y}$), we consider componentwise 
comparison. 
For comparing functions $f,g$ with the same domains, we write $f\leq g$ 
if $f(x)\leq g(x)$ for all $x$ in the domain.

\smallskip\noindent{\em Variables and valuations.}
Throughout the paper we fix a countable set of variables $\vars$. 
We consider some arbitrary but fixed linear order on the set of all variables. 
Hence, given some set of variables $V$ we can enumerate its members in 
ascending order (w.r.t. the fixed ordering) and write 
$V=\{x_1,x_2,x_3,\dots\}$. 

\smallskip\noindent{\em Affine expressions.}
An \emph{affine expression} over the set of variables $\{x_1,\dots,x_n\}$ is an 
expression of the form $d+\sum_{i=1}^{n}a_i
x_i$, where  $d,a_1,\dots,a_n$ are real-valued
constants. Each affine expression $E$ over $\{x_1,\dots,x_n\}$ 
determines a function which for each $m$-dimensional vector $\vec{x}$, where $m\geq n$,  
returns a number resulting from substituting each $x_i$ in $E$ by $\vec{x}[i]$. 
Slightly abusing our notation, we denote this function also by $E$ and the 
value of this function on argument $\vec{x}$ by $E(\vec{x})$. A function of the form $E(\vec{x})$ for some affine expression $E$ is called affine.

\smallskip\noindent{\em Linear constraint, assertion, predicates.}
We use the following nomenclature:
\begin{compactitem}
\item {\em Linear Constraint.} A \emph{linear constraint} is a formula of the
form $\psi$ or $\neg\psi$, where  $\psi$ is a non-strict inequality
between affine expressions.
\item {\em Linear Assertion.} A \emph{linear assertion} is a finite conjunction
of linear constraints.
\item {\em Propositionally Linear Predicate.}
A  \emph{propositionally linear predicate} (PLP) is a finite disjunction of
linear assertions.
\end{compactitem}

\noindent{\em Arity and satisfaction of PLP.}
For a PLP $\varphi$ we denote by $\vars(\varphi)$ the set of all variables that 
appear in $\varphi$. As noted above, we stipulate that $\vars(\varphi)=\{x_1,\dots,x_{n(\varphi)}\}$ for some $n(\varphi)\in \Nset$. A vector $\vec{x}$ of dimension $m\geq n(\varphi)$ \emph{satisfies} $\varphi$, we write 
$\vec{x}\models\phi$, if 
 the arithmetic formula obtained by 
substituting each occurrence of a variable $x_i$ in $\varphi$ by $\vec{x}[i]$ is 
valid. We denote $\sat{\phi} = \{\vec{x}\in \Rset^m\mid m\geq n(\varphi) \wedge \vec{x} \models\phi\}$ 
and $\exsat{\phi} = \sat{\phi}
\cap \Rset^{n(\varphi)}$.

\subsection{Syntax of Affine Probabilistic Programs (\APP s)}\label{subsec:syntax}

\noindent{\em The Syntax.}
We consider the standard syntax for affine probabilistic programs,
which encompasses basic programming mechanisms such as assignment statement 
(indicated by `:='), while-loop, if-branch. We also consider basic probabilistic mechanisms 
such as probabilistic branch (indicated by `prob') and random sampling (e.g. 
$x:=\textbf{sample(}\mathrm{Uniform}[-2,1]\textbf{)}$ assigns to $x$ a random 
number 
uniformly sampled from interval $[-2,1]$). We also allow constructs for 
(demonic) non-determinism, in particular 
non-deterministic branching indicated by `\textbf{if }$\star$ \textbf{then...}' construct and non-deterministic assignment.  
Variables (or identifiers) of a probabilistic program are of \emph{real} type, i.e., 
values of the variables are real numbers. 
We allow only affine expressions in test statements and in the right-hand sides of assignments.
We also assume that assume that each \APP{} $\program$ is preceded by an initialization preamble 
in 
which 
each variable appearing in $\program$ is assigned some concrete number.
Due to space restrictions, details (such as grammar) are relegated to the Appendix.
For an example see Figure~\ref{fig:invariant-running}.
We refer to this class of affine probabilistic programs as \APP s.

\subsection{Semantics of Affine Probabilistic Programs}\label{subsec:semantics}

We now formally define the semantics of \APP's.
In order to do this, we first recall some fundamental concepts from probability
theory.

\smallskip\noindent{\em Basics of Probability Theory.}
The crucial notion is the one of a probability space. A probability space is a triple
$(\Omega,\mathcal{F},\probm)$, where $\Omega$ is a non-empty set (so called
\emph{sample space}), $\mathcal{F}$ is a \emph{sigma-algebra} over $\Omega$,
i.e. a collection of subsets of $\Omega$ that contains the empty set
$\emptyset$, and that is closed under complementation and countable unions, and
$\probm$ is a \emph{probability measure} on $\mathcal{F}$, i.e., a function
$\probm\colon \mathcal{F}\rightarrow[0,1]$ such that
\begin{compactitem}
\item $\probm(\emptyset)=0$,
\item for all $A\in \mathcal{F}$ it holds $\probm(\Omega\smallsetminus
A)=1-\probm(A)$, and
\item for all pairwise disjoint countable set sequences $A_1,A_2,\dots \in
\mathcal{F}$ (i.e., $A_i \cap A_j = \emptyset$ for all $i\neq j$)
we have $\sum_{i=1}^{\infty}\probm(A_i)=\probm(\bigcup_{i=1}^{\infty} A_i)$.
\end{compactitem}

\noindent{\em Random variables and filtrations.}
A \emph{random variable} in a probability space $(\Omega,\mathcal{F},\probm)$ is
an $\mathcal{F}$-measurable function $R\colon \Omega \rightarrow \Rset \cup
\{\infty\}$, i.e.,
a function such that for every $a\in \Rset \cup \{ \infty\}$ the set
$\{\omega\in \Omega\mid R(\omega)\leq a\}$ belongs to $\mathcal{F}$.
We denote by $\expv[R]$ the \emph{expected value} of a random variable $X$~(see \cite[Chapter 5]{Billingsley:book}
for a formal definition). 
A \emph{random vector} in $(\Omega,\mathcal{F},\probm)$ is a vector whose every component is a random 
variable in this probability space. A \emph{stochastic process} in a 
probability space $(\Omega,\mathcal{F},\probm)$ is an infinite sequence of 
random vectors in this space.
We will also use random variables of the form $R\colon\Omega \rightarrow S$ for some finite 
set $S$, which is easily translated to the variables above.
A \emph{filtration} of a sigma-algebra $\mathcal{F}$ is a
sequence $\{\mathcal{F}_i \}_{i=0}^{\infty}$ of $\sigma$-algebras 
such that $\mathcal{F}_0 \subseteq \mathcal{F}_1 \subseteq \cdots \subseteq
\mathcal{F}_n \subseteq \cdots \subseteq \mathcal{F}$.

\emph{Distributions.} We assume the standard definition of a probability 
distribution specified by a cumulative distribution 
function~\cite{Billingsley:book}. We denote by $\mathcal{D}$ be a set of 
probability distributions on 
real numbers, both discrete and continuous.

\smallskip\noindent{\em Probabilistic Control Flow Graphs.}
The semantics can be defined as the semantics of an uncountable state-space
Markov decision process (MDP) (uncountable due to real-valued variables).
We take an operational approach to define the semantics, and associate to each program 
a certain stochastic process~\cite{SriramCAV,HolgerPOPL,Kozen:prob-semantics}.
To define this process, we first define so called 
\emph{probabilistic control flow graphs}~\cite{CFG16:positivstellensatz-arxiv}.

\smallskip
\begin{definition}
\label{def:stochgame}
A \emph{probabilistic control flow graph (pCFG)} is a tuple
$\pCFG=(\locs,\pvars,\locinit,\vecinit,\transitions,\probdist,\guards)$,
where
\begin{compactitem}
\item $\locs$ is a finite set of \emph{locations} partitioned into three 
pairwise
disjoint subsets  $\locs_N$, $\locs_P$, and $\locs_D$ of non-deterministic, 
probabilistic, and deterministic locations;
\item $\pvars=\{x_1,\dots,x_{|\pvars|}\}$ is a finite set of \emph{program 
variables} (note that $\pvars \subseteq \vars$) ;
\item $\locinit$ is an initial location and $\vecinit$ is an initial 
\emph{assignment vector};
\item $\transitions$ is a transition relation, whose members are tuples of
the form $(\ell,i,\up,\ell')$, where $\ell$ and $\ell'$ are source and target
program locations, respectively, $1\leq i \leq |\pvars|$ is a \emph{target variable index}, and $u$ 
is an 
\emph{update element}, which can 
be one of the following mathematical objects: %
(a)~an affine function $u\colon \Rset^{|\pvars|}\rightarrow \Rset$;
(b)~a distribution $d\in \mathcal{D}$; or
(c)~a set $R\subseteq \Rset$.
\item $\probdist=\{\prob_{\ell}\}_{\ell \in \locs_P}$ is a collection of
probability distributions, where each $\prob_{\ell}$ is a discrete probability
distribution on the set of all transitions outgoing from~$\ell$.
\item $\guards$ is a function assigning a propositionally linear predicate
(a \emph{guard}) over $\pvars$ to each transition outgoing from a deterministic 
location.
\end{compactitem}

We assume that each location has at least one outgoing transition.
Also, for every deterministic location $\ell$ we assume the following: if
$\tau_1,\dots,\tau_k$ are all transitions outgoing from $\ell$, then $G(\tau_1)
\vee \dots \vee G(\tau_k) \equiv \mathit{true}$ and $G(\tau_i) \wedge G(\tau_j)
\equiv \mathit{false}$ for each $1\leq i < j \leq k$. Moreover, for each 
distribution $d$ appearing in the 
pCFG we assume the following features are known: expected value $\expv[d]$ of 
$d$ and a 
single-variable PLP $\varphi_d$ such that the \emph{support} of $d$ (i.e. the 
smallest 
closed set of real numbers whose complement has probability zero 
under $d$)\footnote{In particular, a support of a \emph{discrete} probability 
	distribution $d$ is simply the at most countable set of all points on a 
	real 
	line that have positive probability under $d$.} satisfies 
$\support(d)\subseteq\exsat{\varphi_d}$. Finally, we assume that for each 
transition $(\loc,j,u,\loc')$ such that $u$ is a set the location $\loc$ is 
deterministic. This is just a technical assumption yielding no loss of 
generality, and it somewhat simplifies notation.
\end{definition}

\smallskip\noindent{\em Configurations.}
A \emph{configuration} of a pCFG $\pCFG$ is a tuple $(\ell,\vec{x})$,
where $\ell$ is a location of $\pCFG$ and $\vec{x}$ is an 
$|\pvars|$-dimensional vector.
We say that a transition $\tau$ is \emph{enabled} in a configuration
$(\ell,\vec{x})$ if $\ell$ is the source location of $\tau$ and in addition,
${\vec{x}}\models G(\tau)$ provided that $\ell$ is deterministic. A 
configuration 
$(\loc,\vec{x})$ is non-deterministic/probabilistic/deterministic if $\loc$ is 
non-deterministic/probabilistic/deterministic, respectively.

\smallskip\noindent{\em Executions and reachable configurations.}
A \emph{finite path} (or
\emph{execution fragment}) in $\pCFG$ is a finite sequence of
configurations $(\ell_0,\vec{x}_0)\cdots(\ell_k,\vec{x}_k)$ such that 
for each
$0 \leq i < k$ there is a transition  $(\ell_i,j,\up,\ell_{i+1})$ enabled in
$(\ell_i,\vec{x}_i)$ such that $\vec{x}_{i+1}=\vec{x}_i(j\leftarrow a)$ where $a$ 
satisfies one of the following:
\begin{compactitem}
\item $\up$ is a function $f\colon\Rset^{|X|}\rightarrow\Rset$ and
$a=f(\vec{x}_{i})$;
\item $\up$ is an integrable\footnote{A distribution on some numerical domain 
is integrable if its expected value exists and is finite. In particular, each 
Dirac distribution is integrable.} distribution $d$ and 
$a\in \support(d)$; or
\item $\up$ is a set and $a\in \up$.
\end{compactitem}
A \emph{run} (or \emph{execution}) in
$\pCFG$ is an infinite sequence of configurations whose every finite
prefix is a finite path.
A configuration $(\loc,\vec{x})$ is {\em reachable} from the initial 
configuration
$(\locinit,\vecinit)$
if there is a finite path starting in $(\locinit,\vecinit)$ that ends in
$(\loc,\vec{x})$.

\smallskip\noindent{\em Schedulers.}
Due to the presence of non-determinism and probabilistic choices, a pCFG
$\pCFG$ may represent a multitude of possible behaviours. The probabilistic
behaviour of $\pCFG$ can be captured by constructing a suitable
probability measure over the set of all its runs. Before this can be
done, non-determinism in $\pCFG$ needs to be resolved. This is done using the 
standard notion of a \emph{scheduler}.

\smallskip
\begin{definition}[Schedulers]
\label{def:schedulers}
A scheduler in an pCFG $\pCFG$ is a tuple $\sigma=(\sigmat,\sigmaa)$,
where
\begin{compactitem}
\item $\sigmat$ (here '$t$' stands for 'transition') is a function assigning 
to every finite path that ends in a non-deterministic configuration 
$(\loc,\vec{x})$ a probability distribution on transitions outgoing from 
$\loc$; 
and
\item $\sigmaa$ (here '$a$' stands for 'assignment') is a function which takes 
as an argument a finite path ending in a deterministic configuration in which 
some transition 
$(\loc,j,u,\loc')$ with $u$ being a set is enabled, and for such a path it 
returns a probability distribution on $u$.
\end{compactitem}
\end{definition}

\smallskip\noindent{\em Stochastic process.}
A pCFG $\pCFG$ together with a scheduler $\sigma$ can be seen as a stochastic 
process which produces a random run 
$(\loc_0,\vec{x}_0)(\loc_1,\vec{x}_1)(\loc_2,\vec{x}_2)\cdots$. The evolution 
of 
this process can be informally described as follows: we start in the initial 
configuration, i.e. $(\loc_{0},\vec{x}_0)=(\locinit,\vecinit)$. 
Now assume that $i$ 
steps have elapsed, i.e. a finite path 
$(\loc_0,\vec{x}_0)(\loc_1,\vec{x}_1)\cdots(\loc_i,\vec{x}_i)$ has already 
been produced. Then
\begin{compactitem}
\item
A transition $\tau=(\loc,j,u,\loc')$ enabled in 
$(\lr{}{i},\vr{}{i})$ is chosen as follows:

\begin{compactitem}
\item
If $\ell_i$ is non-deterministic then $\tau$ is chosen randomly
according to the distribution specified
by scheduler
$\sigma$, i.e. according to the distribution 
$\sigmat((\loc_0,\vec{x}_0)(\loc_1,\vec{x}_1)\cdots(\loc_i,\vec{x}_i))$.

\item
If $\ell_i$ is probabilistic, then $\tau$ is chosen randomly according to the
distribution $\probdist_{\ell_i}$.
\item
If $\ell_i$ is deterministic, then by the definition of a pCFG there is exactly
one enabled transition outgoing from $\ell_i$, and this transition is chosen as
$\tau$.
\end{compactitem}
\item
Once $\tau$ is chosen as above, we put $\loc_{i+1}=\loc'$. Next, we put 
$\vec{x}_{i+1}=\vec{x}_{i}(j\leftarrow a)$, where $a$ chosen as follows:
\begin{compactitem}
\item If $u$ is a function $u\colon
\Rset^{|\pvars|}\rightarrow \Rset$, then $a=f(\vec{x}_{i})$.
\item If $\up$ is a distribution $d$, then $a$ 
is sampled from $d$.
\item If $\up$ is a set, then $a$ is sampled from a 
distribution 
$\sigmaa((\loc_0,\vec{x}_0)(\loc_1,\vec{x}_1)\cdots(\loc_i,\vec{x}_i))$.
\end{compactitem}
\end{compactitem}

The above intuitive explanation can be formalized by showing that
each pCFG $\pCFG$ together with a scheduler $\sigma$ uniquely determines a 
certain 
probabilistic 
space $(\Omega,\natfilt,\probm^{\sigma})$ in which $\Omega$ is a set of all 
runs in $\pCFG$, and a stochastic process 
$\pCFG^{\sigma}=\{\cfg{\sigma}{i}\}_{i=0}^{\infty}$ in this space 
such 
that for each $\run\in \Omega$ we have that $\cfg{\sigma}{i}(\run)$ is the 
$i$-th configuration on run $\run$ (i.e., $\cfg{\sigma}{i}$ is a random vector 
$(\lr{\sigma}{i},\vr{\sigma}{i})$ with  $\lr{\sigma}{i}$ taking 
values in $\locs$ and $\vr{\sigma}{i}$ being a random vector of dimension 
$|\pvars|$ consisting of real-valued random variables). The sigma-algebra 
$\natfilt$ is the smallest (w.r.t. inclusion) sigma algebra under which all 
the functions $\cfg{\sigma}{i}$, for all $i\geq 0$ and all schedulers $\sigma$, 
are $\natfilt$-measurable (a 
function $f$ returning vectors is $\natfilt$-measurable if for all 
real-valued vectors 
$\vec{y}$ of appropriate dimension the set $\{\omega\in \Omega\mid 
f(\omega)\leq \vec{y}\}$ belongs to $\natfilt$). The 
probability 
measure 
$\probm^{\sigma}$ is such that for each $i$, the distribution of 
$\cfg{\sigma}{i}$ reflects the aforementioned way in which runs are randomly 
generated. The formal construction 
of $\natfilt$ and $\probm^{\sigma}$ is standard~\cite{Billingsley:book} and 
somewhat 
technical, hence we omit it. 
We denote by $\E^\sigma$ the expectation operator in probability space 
$(\Omega,\natfilt,\probm^{\sigma})$.
The translation from probabilistic programs to the corresponding pCFG is 
standard~\cite{CFNH16:prob-termination}, and the details are presented 
in the appendix. We point out that the 
construction produces pCFGs with a property that only transitions outgoing from 
a deterministic state can update program variables. All other transitions are 
assumed to be of the form $(\ell,1,\id_1,\ell')$ for some locations 
$\ell,\ell'$, where $\id_1(\vec{x})=\vec{x}[1]$ for all $\vec{x}$. We use this 
to simplify notation.
An illustration of a pCFG is given in Figure~\ref{fig:invariant-running}.

\subsection{Almost-Sure and Probabilistic Termination}
We consider computational problems related to the basic liveness properties of 
\APP{s}, namely \emph{termination} and its generalization, \emph{reachability}. 

\smallskip\noindent{\em Termination, reachability, and termination time.}
In the following, consider an \APP{} $P$ and its associated pCFG $\pCFG_P$. 
We say that a run $\run$ of $\pCFG_P$ \emph{reaches} a set of configurations 
$C$ if it  contains a configuration from $C$. A run \emph{terminates} if it reaches a 
configuration whose first component (i.e. a location of $\pCFG_P$) is the 
location
$\loc_{P}^{\lout}$ corresponding to the value of the program counter after 
executing $P$. 
To each set of configurations $C$ we can assign a random variable $\treach{C}$ 
such that for each run $\run$ the value $\treach{C}(\run)$ represents the first 
point in time when 
the current configuration on $\run$ is in $C$. If a run $\run$ does \emph{not} 
reach a set $C$, then 
$\treach{C}(\run)=\infty$. We call $\treach{C}$ the 
\emph{reachability time} of $C$. In particular, if $C$ is the set 
of all configurations $(\loc,\vec{x})$ such that $\loc=\loc^\lout_P$ (the 
terminal location of $\pCFG_P$), then $\treach{C}$ is 
called a \emph{termination time}, as it returns the number of steps after which 
$\run$ terminates. Since termination time is an important concept on its own, 
we use a special notation $\ttime$ for it.
Since a probabilistic program may exhibit more than one run, we are interested 
in probabilities of runs that terminate or reach some set of configurations. 
This gives rise to the following fundamental computational problems regarding 
termination:

\begin{compactenum}
\item \emph{Almost-sure termination:} A probabilistic program $P$ is 
almost-surely (a.s.) 
terminating if under each scheduler $\sigma$ it holds that 
$\probm^{\sigma}(\{\run\mid \run \text{ terminates}\}) = 1$, or equivalently, 
if for each $\sigma$ it holds $\probm^{\sigma}(\ttime<\infty)=1$. In 
almost-sure termination question for $P$ we aim to prove that $P$ 
is 
almost-surely terminating.
\item \emph{Probabilistic termination:} In probabilistic termination question 
for $P$ we aim to compute a \emph{lower bound on the probability of 
termination}, i.e. a bound $b\in[0,1]$ such that for each scheduler 
$\sigma$ it holds $\probm^{\sigma}(\{\run\mid \run \text{ terminates}\}) \geq b$ 
(or equivalently $\probm^{\sigma}(\ttime<\infty)\geq b$).
\end{compactenum}

We also define corresponding questions for the more general reachability 
concept. 
 
\begin{compactenum}
\item \emph{Almost-sure reachability:} For a set $C$ of configurations of a
probabilistic program $P$, prove (if possible) that under each scheduler 
$\sigma$ it holds that
$\probm^{\sigma}(\treach{C}<\infty)=1$.
\item \emph{Probabilistic reachability:} For a set $C$ of configurations of 
a probabilistic program $P$, compute a bound $b\in[0,1]$ such that for each scheduler
$\sigma$ it holds $\probm^{\sigma}(\treach{C}<\infty)\geq b$.
\end{compactenum}

Since termination is a special case of reachability, each solution to the 
almost-sure or probabilistic reachability questions provides solution for the 
corresponding termination questions.

\lstset{language=affprob}
\lstset{tabsize=3}
\newsavebox{\invrun}
\begin{lrbox}{\invrun}
\begin{lstlisting}[mathescape]
$x:=10$
while $x\geq 1$ do
	if prob(0.75) then $x:=x-1$	else $x:=x+1$ 
	fi
od
\end{lstlisting}
\end{lrbox}
\begin{figure}[t]
\centering
\usebox{\invrun}

\begin{tikzpicture}[x = 1.8cm]

\node[det] (while) at (1.5,0)  {$\loc_0$};
\node[ran] (prob) at (3,0) {$\loc_1$};
\node[det] (fin) at (0,0) {$\loc_2$};
\draw[tran] (while) to node[font=\scriptsize,draw, fill=white, 
rectangle,pos=0.5] {$x<1$} (fin);
\draw[tran, loop, looseness = 5, in =-65, out = -115] (fin) to (fin);
\draw[tran] (while) to node[font=\scriptsize,draw, fill=white, 
rectangle,pos=0.5] {$x\geq 1$} (prob);

\node (dum1) at (0,0.8) {};
\node (dum2) at (0,-0.8) {};

\draw[tran] (prob) -- node[font=\scriptsize,draw, fill=white, 
rectangle,pos=0.5, inner sep = 1pt] {$\frac{3}{4}$} (prob|-dum1) -- node[auto] 
{x:=x-1}
(while|-dum1)--(while);
\draw[tran] (prob) -- node[font=\scriptsize,draw, fill=white, 
rectangle,pos=0.5, inner sep = 1pt] {$\frac{1}{4}$} (prob|-dum2) 
-- node[auto,swap] {x:=x+1} (while|-dum2)--(while);
\end{tikzpicture}
\caption{An \APP{} modelling an asymmetric 1-D random walk and the associated 
pCFG. Probabilistic locations are depicted by circles, with probabilities given 
on outgoing 
transitions. Transitions are labelled by their effects. Location $\loc_0$ is 
initial and $\loc_2$ is terminal.}
\label{fig:invariant-running}
\end{figure}
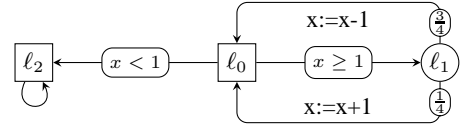

\section{Invariants and Ranking Supermartingales}\label{sec:invm}

In this section we recall known methods and constructs for solving the 
qualitative termination and reachability questions for \APP{s}, namely  
linear invariants and ranking supermartingales. 
We also demonstrate that these methods are not sufficient to address the 
quantitative variants of these questions (i.e., probabilistic termination).
In order to discuss the necessary concepts, we recall the basics of 
martingales, which is relevant for both this and subsequent sections. 
\subsection{Pure Invariants}

Invariants are a vital element of many program analysis techniques. 
Intuitively, invariants are maps assigning to each 
program location $\loc$ of some pCFG a predicate which is guaranteed to hold 
whenever $\loc$ is 
entered. To avoid confusion with stochastic invariants, 
that we introduce later, we call these standard invariants \emph{pure invariants}.

\smallskip
\begin{definition}[Linear Predicate Map (LPM) and Pure Invariant] We define the 
following:
\begin{compactenum}
\item
A \emph{linear predicate map (LPM)} for an \APP{} $P$ is a function $I$ 
assigning to each location $\loc$ of the pCFG $\pCFG_P$ a propositionally 
linear predicate $I(\loc)$ over the set of program variables of $P$.
\item
A \emph{pure linear invariant} (or just a pure invariant) for an \APP{} $P$ is 
a linear predicate map $\inv$ for $P$ with
the following property: for each location $\loc$ of $\pCFG_P$ and each finite 
path $(\loc_0,\vec{x}_0),\cdots,(\loc_n,\vec{x}_n)$ such that 
$(\loc_0,\vec{x}_0)=(\locinit,\vecinit)$ and $\loc_n = \loc$ it holds
$\vec{x}_n\models \inv(\loc)$.
\end{compactenum}
\end{definition}

\subsection{Supermartingales}
(Super)martingales, are a standard tool of 
probability theory apt for analyzing probabilistic objects arising in computer 
science, from automata-based models~\cite{BKKNK:pMC-zero-reachability} to 
general probabilistic 
programs~\cite{SriramCAV,HolgerPOPL,CFNH16:prob-termination,CFG16:positivstellensatz-arxiv,BEFH16:doob-decomposition}.

Let us first recall basic
definitions and results related to supermartingales, which we need in our 
analysis.

\smallskip\noindent{\bf Conditional Expectation.} 
Let $(\Omega,\mathcal{F},\probm)$ be a probability space, 
$X\colon\Omega\rightarrow 
\Rset$ an $\mathcal{F}$-measurable function, and $\mathcal{F}'\subseteq 
\mathcal{F}$ sub-sigma-algebra of $\mathcal{F}$. The \emph{conditional 
expectation} of $X$ given $\mathcal{F}'$ is an $\mathcal{F}'$-measurable random 
variable denoted by $\E[X| \mathcal{F}']$ which satisfies, for each set $A\in 
\mathcal{F}'$, the following: 
\begin{equation}
\label{eq:cond-exp}
\E[X\cdot 1_A] = \E[\E[X|\mathcal{F}]\cdot 1_A],
\end{equation}
where $1_A \colon \Omega\rightarrow \{0,1\}$ is an \emph{indicator function} of 
$A$, i.e. function returning $1$ for 
each $\omega\in A$ and $0$ for each $\omega\in \Omega\setminus A$. Note that 
the left hand-side of~\eqref{eq:cond-exp} intuitively represents the expected 
value 
of $X(\omega)$ with domain restricted to $A$.

Recall that in context of probabilistic programs we work with probability 
spaces of the form $(\Omega,\natfilt,\probm^\sigma)$, where $\Omega$ is a 
set of runs in some $\pCFG$ and $\mathcal{F}$ is (the smallest) sigma-algebra 
such that all the functions $\cfg{\sigma}{i}$, where $i\in \Nset_0$ and 
$\sigma$ is a scheduler, are $\natfilt$-measurable. In such a setting we can 
also consider sub-sigma-algebras $\natfilt_i$, $i\in \Nset_0$, of 
$\natfilt$, where $\natfilt_i$ is the smallest sub-sigma-algebra of 
$\natfilt$ such that all the functions $\cfg{\sigma}{j}$, $0\leq j \leq 
i$, are $\natfilt_i$-measurable. Intuitively, each set $A$ belonging to such 
an $\natfilt_i$ consists of runs whose first $i$ steps satisfy some 
property, and the probability space $(\Omega,\natfilt_i,\probm^\sigma)$ 
allows us to reason about probabilities of certain events happening in the 
first 
$i$ steps of program execution. 
Then, for each $A\in \natfilt_i$, the 
value $\E[\E[X|\natfilt_i]\cdot 1_A]$ represents the expected value of 
$X(\run)$ for the randomly generated run $\run$ provided that we restrict to 
runs whose
prefix of length $i$ satisfies the property given by $A$. 
Note that the sequence $\natfilt_0,\natfilt_1,\natfilt_2,\dots$ forms a 
filtration of $\natfilt$, which we call a \emph{canonical filtration}.

\smallskip
\begin{definition}[Supermartingale]
Let $(\Omega,\mathcal{F},\probm)$ be a probability space and 
$\{\mathcal{F}_i\}_{i=0}^{\infty}$ a filtration of $\mathcal{F}$. A sequence of 
random variables $\{X_i\}_{i=0}^{\infty}$ is a \emph{supermartingale} w.r.t. 
filtration $\{\mathcal{F}_i\}_{i=0}^{\infty}$ if it satisfies these conditions:
\begin{compactenum}
\item  The process $\{X_i\}_{i=0}^{\infty}$ is \emph{adapted} to 
$\{\mathcal{F}_i\}_{i=0}^{\infty}$, i.e. for all $i\in \Nset_0$ it holds that 
$X_i$ is $\mathcal{F}_i$-measurable.
\item For all $i\in \Nset_0$ it holds $\E[|X_i|]<\infty$.
\item For all $i\in \Nset_0$ it holds 
\begin{equation}
\label{eq:supermart-def}
\E[X_{i+1}|\mathcal{F}_i] \leq X_i.
\end{equation}

A supermartingale $\{X_i\}_{i=0}^{\infty}$ has $c$-bounded differences, where 
$c\geq 0$, if $|X_{i+1}-X_i|<c$ for all $i\in \Nset_0$
\end{compactenum}
\end{definition}

Intuitively, a supermartingale is a stochastic process whose average value is 
guaranteed not to rise as time evolves, even if some information on the past 
evolution of the process is revealed. 
We often need to work with supermartingales whose value 
is guaranteed to decrease on average, until a certain condition is 
satisfied. The point in time in which such a condition is satisfied is called a 
\emph{stopping time}.

\smallskip
\begin{definition}[Stopping time]
Let $(\Omega,\mathcal{F},\probm)$ be a probability space and 
$\{\mathcal{F}_i\}_{i=0}^{\infty}$ a filtration. A random variable $\stime 
\colon 
\Omega\rightarrow \Nset_0$ is called a \emph{stopping time} w.r.t. 
$\{\mathcal{F}_i\}_{i=0}^{\infty}$ if %
for all $j\in \Nset_0$ the set $\{\omega\in \Omega\mid \stime(\omega)\leq j\}$ 
belongs to $\mathcal{F}_j$.
\end{definition}

In particular, for each set of configurations $C$ the reachability time 
$\treach{C}$ of $C$ is a stopping time w.r.t. the canonical filtration, since 
at each time $j$ we can decide whether $\treach{C}>j$ or not by looking at the 
prefix of a run of length $j$. 
Finally, we recall the fundamental notion of a ranking 
supermartingale.

\smallskip
\begin{definition}[Ranking supermartingale]
\label{def:ranking}
Let $(\Omega,\mathcal{F},\probm)$ be a probability space, 
$\{\mathcal{F}_i\}_{i=0}^{\infty}$ a filtration of $\mathcal{F}$, $\stime$ a 
stopping time w.r.t. that filtration, and $\eps\geq 
0$. 
A supermartingale $\{X_i\}_{i=0}^{\infty}$ (w.r.t. 
$\{\mathcal{F}_i\}_{i=0}^{\infty}$) is \emph{$\eps$-decreasing} until 
$\stime$ 
if it 
satisfies 
the following additional condition: for all $i\in \Nset_0$ it holds
\begin{equation}
\label{eq:ranking-sup}
\E[X_{i+1}|\mathcal{F}_i] \leq X_i - \eps\cdot\vec{1}_{\stime > i}.
\end{equation}

Further, $\{X_i\}_{i=0}^{\infty}$ is an \emph{$\eps$-ranking} supermartingale 
($\eps$-RSM) for $\stime$ if it 
is $\eps$-decreasing until $\stime$ and for 
each $\omega\in\Omega, j\in\Nset_0$ it holds $\stime(\omega)>j \Rightarrow 
X_j(\omega)\geq 0$.
\end{definition}
Intuitively, if $\stime$ is the reachability time $\treach{C}$ of some set $C$, 
then the previous definition requires that an $\eps$-ranking supermartingale 
must decrease by at least $\eps$ on average up to the point when $C$ is 
reached for a first time. After that, it must not increase (on average).
The above definition is a bit more general than the standard one in the
literature as we also consider reachability as opposed to only termination.

\smallskip\noindent{\bf Martingales in Program Analysis.}
In the context of \APP{} analysis, we consider a special type of 
supermartingales given as functions of the current values of program variables. 
In this paper we focus on the case when these functions are \emph{linear}.

\smallskip
\begin{definition}[Linear Expression Map]
A \emph{linear expression map (LEM)} for an \APP{} $P$ is a function $\lem$ 
assigning to each program location $\loc$ of $\pCFG_P$ an affine expression 
$\lem(\ell)$ over the program variables of $P$.
\end{definition}

Each LEM $\lem$ and location $\loc$ determines an affine function $\lem(\loc)$ 
which takes as an argument an $n$-dimensional vector, where $n$ is the number 
of distinct variables in $P$. We use $\lem(\loc,\vec{x})$ as a shorthand 
notation for $\lem(\loc)(\vec{x})$. 
Martingales for \APP{} analysis are defined via a standard notion of 
pre-expectation~\cite{SriramCAV}. Intuitively, a pre-expectation of $\lem$ is a 
function which for each configuration $(\loc,\vec{x})$ returns the maximal
expected value of $\lem$ after one step is made from this configuration, where 
the maximum is taken over all possible non-deterministic choices.

\smallskip
\begin{definition}[Pre-Expectation]
Let $P$ be an \APP{} such that $\pCFG_P = 
(\locs,\pvars,\locinit,\vecinit,\transitions,\probdist,\guards)$ and let $\lem$ 
a 
linear expression map 
for $P$. 
The 
pre-expectation of $\lem$ is a function $\preexp{\lem}\colon \locs\times 
\Rset^{|\pvars|} \rightarrow \Rset$ defined as follows:
\begin{compactitem} %
\item 
if $\loc$ is a probabilistic location, then
$$\preexp{\lem}(\loc,\vec{x}):=\sum_{(\loc,1,\id_1,\loc')\in\transitions} 
Pr_{\loc}\left((\loc,1,\id_1,\loc')\right)\cdot
 \lem(\loc',\vec{x});$$
\item 
 if $\loc$ is a non-deterministic location, then
$$\preexp{\lem}(\loc,\vec{x}):=
\max_{(\loc,1,\id_1,\loc')\in\transitions}\lem(\loc',\vec{x});$$

\item 
if $\loc$ is a deterministic location, then for each $\vec{x}$ the value 
$\preexp{\lem}(\loc,\vec{x})$ is determined as follows: there is exactly one 
transition
$\tau=(\loc,j,\up,\loc')$ such that $\vec{x}\models G(\tau)$. We distinguish 
three cases:
\begin{compactitem}
\item If $\up\colon \Rset^{|\pvars|}\rightarrow \Rset$ is a function, then 
$$\preexp{\lem}(\loc,\vec{x}):= \lem(\loc',\vec{x}(j\leftarrow \up(\vec{x}))).$$
\item If $\up$ is a distribution $d$, then $$ \preexp{\lem}(\loc,\vec{x}):= 
\lem(\loc',\vec{x}(j\leftarrow \E[d])),$$ where $\E[d]$ is the expected value of the 
distribution $d$.
\item 
If $\up$ is a set, then $$ \preexp{\lem}(\loc,\vec{x}):= \max_{a\in\up}
\lem(\loc',\vec{x}(j \leftarrow a)).$$
\end{compactitem}
\end{compactitem}
\end{definition}

\smallskip
\begin{definition}(Linear Ranking Supermartingale)
\label{def:lrsm}
Let $P$ be an \APP{} such that $\pCFG_P = 
(\locs,\pvars,\locinit,\vecinit,\transitions,\probdist,\guards)$, let $\inv$ be 
a linear predicate map and let 
$\confset\subseteq \locs \times \Rset^{|\pvars|}$ be some set of 
configurations. 
A linear $\eps$-ranking supermartingale ($\eps$-LRSM) for $\confset$ 
supported by $\inv$ is a 
linear expression map $\lem$ for $P$ such that for all configurations 
$(\loc,\vec{x})$ of $\pCFG_P$ with
$(\loc,\vec{x})\not\in \confset$ and $\vec{x}\models \inv(\loc)$ the following 
two conditions hold:
\begin{compactitem}
\item
$\lem(\loc,\vec{x})\geq 0$
\item 
$\preexp{\lem}(\loc,\vec{x}) \leq \lem(\loc,\vec{x})-\eps$ 
\end{compactitem}
A linear $\eps$-ranking supermartingale supported by $\inv$ has 
$c$-bounded differences if for each $(\loc,\vec{x})$ such that $\vec{x}\models 
\inv(\loc)$ and each configuration $(\loc',\vec{x}')$ such that $(\loc,\vec{x}) 
(\loc',\vec{x}')$ is a path in $\pCFG_P$ it holds 
$|\lem(\loc,\vec{x})-\lem(\loc',\vec{x}')|\leq c$. 
\end{definition}

The relationship between $\eps$-LRSM in \APP{s}, (pure) invariants, and almost-sure termination 
is summarized in the following theorem. 
\smallskip
\begin{theorem}[{\cite[Theorem 1]{CFNH16:prob-termination}}]
\label{thm:old-ranking}
Let $P$ be an \APP{}, $\sigma$ a scheduler, and 
$(\Omega,\natfilt,\probm^\sigma)$ the corresponding probability space. Further, 
let $C$ be the set of terminating configurations of $\pCFG_P$ (i.e., the termination
location is reached), such that there exist an $\eps>0$ and an $\eps$-linear ranking 
supermartingale $\lem$ supported by a pure invariant $I$. 
Then 
\begin{compactenum}
\item
$\probm^{\sigma}(\ttime<\infty)=1$, i.e. termination is ensured almost-surely.
\item
$\E^{\sigma}[\ttime]<\lem(\locinit,\vecinit)/\eps$.
\end{compactenum}
\end{theorem}

The previous result shows that if there exists  an $\eps$-LRSM 
supported by a pure invariant $I$, for $\eps>0$, then under each scheduler 
termination is ensured almost-surely.
We now demonstrate that pure invariants, though effective for almost-sure 
termination, are ineffective to answer probabilistic termination questions.

\lstset{language=affprob}
\lstset{tabsize=3}
\newsavebox{\nonterm}
\begin{lrbox}{\nonterm}
\begin{lstlisting}[mathescape]
$x:=30,y:=20$
while $y\geq 0$ do
	$x:=x+$sample$(\mathrm{Uniform}[-\frac{1}{4},1])$
	$y:=y+$sample$(\mathrm{Uniform}[-1,\frac{1}{4}])$
	while $x \leq 0$ do skip od
od
\end{lstlisting}
\end{lrbox}
\begin{figure}[t]
\usebox{\nonterm}
\begin{tikzpicture}[x = 1.7cm]

\node[det] (while) at (1.5,0)  {$\loc_0$};
\node[det] (p1) at (2.5,0) {$\loc_1$};
\node[det] (p2) at (3.5,0) {$\loc_2$};

\node[det] (lp) at (4.5,0) {$\loc_3$};
\node[det] (fin) at (0,0) {$\loc_4$};
\node (dum1) at (0,0.8) {};
\node (dum2) at (0,-0.8) {};
\draw[tran] (while) to node[font=\scriptsize,draw, fill=white, 
rectangle,pos=0.5] {$y<0$} (fin);
\draw[tran, loop, looseness = 5, in =-65, out = -115] (fin) to (fin);
\draw[tran] (while) -- node[font=\scriptsize,draw, fill=white, 
rectangle,pos=0.5, inner sep = 2pt] {$y\geq0$} node[label={[label 
distance=0.12cm]-90:{x:=...}}] {} (p1);
\draw[tran] (p1) to node[label={[label 
distance=0.12cm]-90:{y:=...}}] {} (p2);
\draw[tran] (lp) to  (p2);
\draw[tran] (p2) -- (p2|-dum2) -- node[font=\scriptsize,draw, fill=white, 
rectangle,pos=0.5, inner sep = 2pt] {$x\leq 0$} (lp|-dum2) -- (lp);
\draw[tran] (p2) -- (p2|-dum1) -- node[font=\scriptsize,draw, fill=white, 
rectangle,pos=0.5, inner sep = 2pt] {$x> 0$} (while|-dum1) -- (while);
\end{tikzpicture}
\caption{A program with infinitely many reachable configurations which 
terminates with high probability, but not almost 
surely, together with a sketch of its pCFG. }\label{fig:nonterm}
\end{figure}
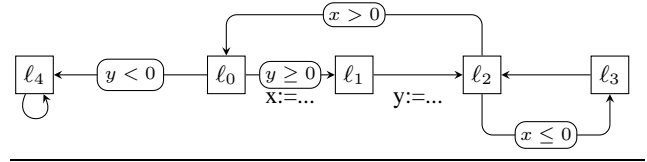

\smallskip
\begin{example}\label{ex:nonterm}
Consider the program in Figure~\ref{fig:nonterm}. 
In each 
iteration of the outer loop each of the variables is randomly modified by 
adding a number drawn from some uniform distribution.
Average increase of $x$ in each iteration is $\frac{3}{8}$, while average decrease 
of $y$ is $-\frac{3}{8}$.  It is easy to see that a program does not 
terminate almost-surely: 
there is for instance a tiny but non-zero probability of $x$ being decremented 
by at least $\frac{1}{8}$ in 
each of the first 240 loop iterations, after which we are stuck in the infinite 
inner 
loop. On 
the other hand, the expectations above show that there is a ``trend'' of $y$ 
decreasing and $x$ increasing, and executions that follow this trend eventually 
decrement $y$ below $0$ without entering the inner loop. Hence, the 
probabilistic intuition tells us that the program terminates with a high 
probability. However, the techniques of this section cannot prove 
this high-probability termination, since existence of an $\eps$-LRSM (with 
$\eps>0$) supported by a pure invariant already implies a.s. termination,
and so no such $\eps$-LRSM can exist for the program. 
\end{example}

In the next section we generalize the notion of pure invariants to stochastic invariants
for probabilistic termination to resolve issues like Example~\ref{ex:nonterm}.

\section{Stochastic Invariants and Probabilistic Termination}
In this section we introduce {\em stochastic invariants}.
Intuitively, stochastic invariants are linear predicate maps extended with 
an upper bound on the probability of their violation.

\smallskip
\begin{definition}[Stochastic Linear Predicate Maps and Invariants]
Stochastic linear predicate maps and stochastic invariants are defined
as follows:
\begin{compactitem}
\item A stochastic linear predicate map (SLPM) for an 
\APP{} $P$ is a pair $(\pinv,p)$ where
$\pinv$ is a linear predicate map and
$p\in[0,1]$ is a probability.
\item A stochastic linear invariant (or just a stochastic invariant) for an 
\APP{} $P$ is an SLPM $(\pinv,p)$ for $P$ with the following property: if we 
denote by $\fail{\pinv}{}$ the set of all runs initiated in 
$(\locinit,\vecinit)$ that reach a configuration of 
the form $(\loc,\vec{x})$ with 
$\vec{x}\not\models\pinv(\loc)$, then for all schedulers $\sigma$ it 
holds 
$\probm^{\sigma}(\fail{\pinv}{}) \leq p$.
\end{compactitem}
\end{definition}

\smallskip
\begin{example}
Consider the \APP{} consisting of a single statement 
$x:=$\textbf{sample}$(\mathrm{Uniform}[0,2])$. Denoting $\loc^{\lin},\loc^{\lout}$ the 
initial and 
terminal 
location of this program, respectively, the stochastic LPM 
$(\pinv,\frac{1}{2})$, where 
$\pinv$ is such that $\pinv(\loc^{\lout})\equiv x\geq 1$ and $\pinv\equiv 
\mathit{true}$ is a stochastic invariant for the program. 
\end{example}

\smallskip
\begin{example}
\label{ex:unproven-invariant}
Consider the example in Figure~\ref{fig:nonterm} and a stochastic LPM 
$(\pinv,p)$ for the program such that $\pinv(\loc_2)\equiv x\geq 1$, 
$\pinv(\loc)\equiv \mathit{true}$ for all the other locations, and 
$p=10^{-5}$. 
Then it is possible to prove that $(\pinv,p)$ is a stochastic invariant for the 
program. 
\end{example}

Before presenting our result related to stochastic invariants, we first present
a technical result. 
Intuitively, the result states that if we have an $\eps$-LRSM for some set of 
configurations $C$ supported by some linear predicate map $\inv$, then we can 
use it to obtain a supermartingale which decreases by at least $\eps$ (on 
average) in each step until we reach either a configuration in $C$ or a 
configuration that does not satisfy $\inv$. In particular, if $\inv$ is a pure 
invariant, the resulting supermartingale decreases until we reach $C$.
This is a result about pure invariants, which we will extend to stochastic 
invariants.

\smallskip
\begin{lemma}
\label{thm:pure-supermart}
Let $P$ be an \APP{} and $\lem$ a linear $\eps$-ranking 
supermartingale for some set $\confset$ of  configurations of $\pCFG_P$ 
supported by $\inv$. 
Let $\neg \inv$ be the set of all configurations $(\loc,\vec{x})$ such that 
$\vec{x}\not\models 
\inv(\loc)$.
Finally, let
$\{X_i\}_{i=0}^{\infty}$ be a stochastic process defined by 
\[
X_i(\run) = \begin{cases}
\lem( 
\cfg{\sigma}{i}(\run)) & \text{ if $\treach{C\cup \neg \inv} \geq i$} \\
X_{i-1}(\run) & \text{otherwise}.
\end{cases}
\]
Then under each scheduler $\sigma$ the stochastic process 
$\{X_i\}_{i=0}^{\infty}$ is an $\eps$-ranking 
supermartingale for $\treach{C\cup \neg \inv}$. 
Moreover, if $\lem$ has $c$-bounded differences, then so has 
$\{X_i\}_{i=0}^{\infty}$.
In  particular, if $\inv$ is a pure invariant of $P$, then 
 $\{X_i\}_{i=0}^{\infty}$ is an $\eps$-ranking 
supermartingale for $\treach{C}$.
\end{lemma}

We now establish a crucial connection between stochastic invariants, linear 
ranking supermartingales, and quantitative reachability (and thus quantitative 
termination).

\smallskip
\begin{theorem}
\label{thm:quantitative-termination}
Let $\{(\pinv_1,p_1),\dots,(\pinv_n,p_n)\}$ be a set of stochastic linear 
invariants for \APP{} $P$, and let $\inv$ be a linear predicate map for $P$ 
such that for each location $\loc$ of $\pCFG_P$ the formula $\inv(\loc)$ 
is entailed by the formula $\pinv_1(\loc)\wedge\dots\wedge\pinv_n(\loc) $. If 
there 
exists a linear 
$\eps$-ranking supermartingale $\lem$ for a set of configurations $C$ such that 
$\lem$ is supported by $\inv$, then under each scheduler $\sigma$ it holds 
$\probm^{\sigma}(\treach{C}<\infty)\geq 1-\sum_{j=1}^{n}p_j$.
\end{theorem}
\begin{proof}
From Lemma~\ref{thm:pure-supermart} it follows that there is an 
$\eps$-ranking supermartingale $\{X_i\}_{i=0}^{\infty}$ for $\treach{C\cup \neg 
\inv}$. We can prove a generalization of Theorem~\ref{thm:old-ranking} for 
other stopping times apart from $\ttime$, which gives us that under each 
scheduler 
$\sigma$ the set of configurations $C \cup \neg \inv$ is reached with 
probability 1, where $\neg \inv$ is the set of all $(\loc,\vec{x})$ such that 
$\vec{x}\not\models\inv(\loc)$. But since each $(\pinv_j,p_j)$ is a stochastic 
invariant, the 
probability that $\neg \pinv_j$ is reached is at most $p_j$ under each 
scheduler. Using union bound the probability of reaching 
$\bigcup_{j=1}^{n}\neg\pinv_j$ is at most $\sum_{j=1}^{n}p_j$, from which the 
result follows.
\end{proof}

\begin{example}
Let $(\pinv,10^{-5})$ be the stochastic invariant from 
Example~\ref{ex:unproven-invariant} (concerning Figure~\ref{fig:nonterm}). For 
the corresponding program we can easily infer a pure invariant $\inv'$ such that
$\inv'(\loc_1)\equiv y \geq 0$, $\inv'(\loc_2)=\inv'(\loc_0)\equiv y \geq -1$ 
and $\inv'(\loc_4)\equiv\inv'(\loc_3)\equiv\mathit{true}$ (actually, standard 
methods would likely infer stronger pure invariants, but $\inv'$ is sufficient 
for the sake of example).
Consider a LEM $\lem$ defined as follows $\lem(\loc_0) = 8y+9 $, 
$\lem(\loc_1)=8y+8$, $\lem(\loc_2)=8y+10$, 
$\lem(\loc_3)=8y+11$ and $\lem(\loc_4)=-1$. Then 
$\lem$ is a $1$-LRSM for the set of terminal configurations supported by LPM
$\inv = \inv'\wedge \pinv$ (where the conjunction is 
locationwise). Now consider a set $\{(\inv',0),(\pinv,10^{-5})\}$. From 
Example~\ref{ex:unproven-invariant} and from the fact that $\inv'$ is a pure 
invariant it follows that both members of the set are stochastic invariants, 
and clearly $\inv'\wedge \pinv$ entails $\inv$. From 
Theorem~\ref{thm:quantitative-termination} 
it follows that the 
program 
terminates with 
probability at least $0.99999$.
\end{example}

Theorem~\ref{thm:quantitative-termination} shows a way in which probabilistic 
reachability and termination properties of \APP{s} can be proved by use of 
ranking supermartingales and stochastic invariants. As highlighted in 
Example~\ref{ex:unproven-invariant}, the crucial question now is 
proving the existence of suitable stochastic invariants for \APP{}. 
While there are various methods of obtaining pure linear 
invariants~\cite{CS02:practical-termination}, e.g. those based on abstract 
interpretation~\cite{DBLP:conf/popl/CousotC77}, 
constraint solving~\cite{DBLP:conf/cav/BradleyMS05} etc., these 
methods do not support reasoning about 
probabilities of a given assertion being satisfied, and thus they are not 
sufficient for obtaining stochastic invariants. In the next section we 
propose a framework for reasoning about stochastic invariants using 
repulsing supermartingales.

\section{Proving Stochastic Invariance with Repulsing Supermartingales}\label{sec:repsm}

Consider that we want to use stochastic invariants to prove that some \APP{} 
$P$ terminates with a high probability, by using 
Theorem~\ref{thm:quantitative-termination}. We need to achieve two things: 
\begin{compactenum}
\item[a.] obtain a linear predicate map $\pinv$ which supports some linear 
ranking supermartingale 
for the termination time $\ttime$ of $P$; and 
\item[b.] obtain an upper bound $p$ on the 
probability that $\pinv$ is violated. 
\end{compactenum}
The part a. is not in any way related to the probability of $p$ being 
satisfied, and 
hence we can aspire to adapt some of the techniques for generation of pure 
invariants. 
The part b. is substantially trickier, since it requires quantitative 
reasoning 
about the highly complex stochastic process $\{\cfg{\sigma}{i}\}_{i=0}^{\infty}$. To achieve 
this task, we introduce a notion of \emph{$\eps$-repulsing} supermartingale. 

\smallskip\noindent{\em Intuitive idea of repulsing supermartingales.}
Intuitively, $\eps$-repulsing supermartingales are again required to decrease 
by 
at least $\eps$ on average 
in every step until some stopping time, e.g. until reaching some set $C$ of 
configurations. But now, instead of requiring the value of the process to be 
non-negative until $C$ reached, we require it to be non-negative 
\emph{upon} reaching $C$. This is because we typically work with repulsing 
supermartingales whose initial value is non-positive. Then, intuitively  an 
$\eps$-repulsing supermartingale is driven 
\emph{away} from 
non-negative values by at least $\eps$ per step, which provides a probabilistic 
argument for showing that the $C$ is reached with small probability. 

\smallskip
\begin{definition}[Repulsing supermartingale]
Let $(\Omega,\mathcal{F},\probm)$ be a probability space, 
$\{\mathcal{F}_i\}_{i=0}^{\infty}$ a filtration of $\mathcal{F}$, $\stime$ a 
stopping time w.r.t. that filtration, and $\eps\geq 
0$. 
A supermartingale $\{X_i\}_{i=0}^{\infty}$ (w.r.t. 
$\{\mathcal{F}_i\}_{i=0}^{\infty}$) is \emph{$\eps$-repulsing} for 
$\stime$ if it is $\eps$-decreasing until $\stime$ and for 
each $\omega\in\Omega, j\in\Nset_0$ it holds $\stime(\omega)=j \Rightarrow 
X_j(\omega)\geq 0$.
\end{definition}

To apply repulsing supermartingales to concrete programs, we again define the 
important special case of {\em linear repulsing supermartingales}.

\smallskip
\begin{definition}[Linear repulsing supermartingale]\label{def:lrepsm}
Let $P$ be an \APP{} such that $\pCFG_P = 
(\locs,\pvars,\locinit,\vecinit,\transitions,\probdist,\guards)$, let $\inv$ be 
a linear predicate map and let 
$\confset\subseteq \locs \times \Rset^{|\pvars|}$ be some set of 
configurations. 
A linear $\eps$-repulsing supermartingale ($\eps$-LRepSM) for a set $\confset$ 
supported by $\inv$ is an LEM $\lem$ for $P$ such that for all configurations 
$(\loc,\vec{x})$ of $\pCFG_P$ such that $\vec{x}\models \inv(\loc)$ the 
following holds
\begin{compactitem}
\item
if $(\loc,\vec{x})\in \confset$, then $\lem(\loc,\vec{x})\geq 0$
\item 
if $(\loc,\vec{x})\not\in \confset$ and $\loc$ is not a terminal location, then
$\preexp{\lem}(\loc,\vec{x}) \leq \lem(\loc,\vec{x})-\eps$,
\end{compactitem}
An $\eps$-LRepSM supported by $\inv$ has 
$c$-bounded differences if for each pair of locations $\loc,\loc'$, each 
transition $\tau$ from $\loc,\loc'$, and each pair of
configurations $(\loc,\vec{x})$, $(\loc',\vec{x}')$ such that 
$\vec{x}\models\inv(\loc)\wedge G(\tau)$ and $(\loc',\vec{x}')$ can be produced 
by performing $\tau$ in $(\loc,\vec{x})$ it holds 
$|\lem(\loc,\vec{x})-\lem(\loc',\vec{x}')|\leq c$.
\end{definition}

\lstset{language=affprob}
\lstset{tabsize=3}
\newsavebox{\nontermnew}
\begin{lrbox}{\nontermnew}
\begin{lstlisting}[mathescape]
$x:=10$
while $x\geq 0$ do
	if prob(0.5) then $x:=x+1$
		else $x:=x-2$
	fi
od
\end{lstlisting}
\end{lrbox}
\begin{figure}[t]

\usebox{\nontermnew}
\begin{center}
\begin{tikzpicture}[x = 1.8cm]

\node[det] (while) at (1.5,0)  {$\loc_0$};
\node[ran] (prob) at (3,0) {$\loc_1$};
\node[det] (fin) at (0,0) {$\loc_2$};
\draw[tran] (while) to node[font=\scriptsize,draw, fill=white, 
rectangle,pos=0.5] {$x<0$} (fin);
\draw[tran, loop, looseness = 5, in =-65, out = -115] (fin) to (fin);
\draw[tran] (while) to node[font=\scriptsize,draw, fill=white, 
rectangle,pos=0.5] {$x\geq 0$} (prob);

\node (dum1) at (0,0.8) {};
\node (dum2) at (0,-0.8) {};

\draw[tran] (prob) -- node[font=\scriptsize,draw, fill=white, 
rectangle,pos=0.5, inner sep = 1pt] {$\frac{1}{2}$} (prob|-dum1) -- node[auto] 
{x:=x-2}
(while|-dum1)--(while);
\draw[tran] (prob) -- node[font=\scriptsize,draw, fill=white, 
rectangle,pos=0.5, inner sep = 1pt] {$\frac{1}{2}$} (prob|-dum2) 
-- node[auto,swap] {x:=x+1} (while|-dum2)--(while);
\end{tikzpicture}
\end{center}
\caption{A probabilistic program example, with the accompanying 
pCFG.}\label{fig:lrepsm}
\end{figure}
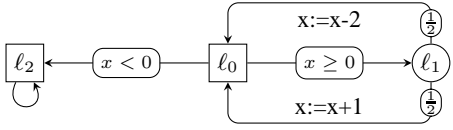

\smallskip
\begin{example}[Illustration of LRepSM]\label{ex:lrepsm}
Consider the program shown in Figure~\ref{fig:lrepsm}, with initial value 
$x:=10$.
Consider a linear predicate map $\pinv$ such that 
$\pinv(\loc_0)\equiv x \leq 500$ and $\pinv(\loc_1)\equiv\pinv(\loc_2)\equiv 
\mathit{true}$.
Consider an LEM $\eta$ that assigns to each pair $(\loc_i,x)$ a value $7\cdot 
x+d_i$, where $d_i$ is the $i$-th component of the ordered tuple 
$(-3499,-3500,-3500)$. 
It is straightforward to verify that $\eta$ is a $1$-LRepSM for $\neg\pinv$ 
supported 
by a trivial pure invariant assigning $\mathit{true}$ to each location.
\end{example}

The connection between $\eps$-LRepSMs and general $\eps$-repulsing 
supermartingales is similar as for their ranking variants 
(Lemma~{\ref{thm:pure-supermart}}). That is, from 
$\eps$-LRepSMs we can obtain a stochastic process which is a supermartingale 
w.r.t. the canonical filtration, which decreases at least by $\eps$ on average 
until the some set $C$ is reached, and upon reaching $C$ its value is 
non-negative.

\smallskip
\begin{lemma}
\label{thm:rep-supermart-connect}
Let $P$ be an \APP{} and $\lem$ an $\eps$-LRepSM 
for some set $\confset$ of  configurations of $\pCFG_P$ 
supported 
by some linear predicate map $\inv$. Let $\neg \inv$ be the set of all 
configurations $(\loc,\vec{x})$ 
such that $\val_\vec{x}\not\models 
\inv(\loc)$.
Finally, let
$\{X_i\}_{i=0}^{\infty}$ be a stochastic process defined by 
\[
X_i(\run) = \begin{cases}
\lem( 
\cfg{\sigma}{i}(\run)) & \text{ if $\treach{C\cup \neg \inv} \geq i$} \\
X_{i-1}(\run) & \text{otherwise}.
\end{cases}
\]
Then under each scheduler 
$\sigma$ the stochastic process 
$\{X_i\}_{i=0}^{\infty}$ is an $\eps$-repulsing supermartingale for 
$\treach{C\cup \neg \inv}$. 
Moreover, if $\lem$ has $c$-bounded differences, then so has 
$\{X_i\}_{i=0}^{\infty}$.
In  particular, if $\inv$ is a pure invariant of $P$, then 
 $\{X_i\}_{i=0}^{\infty}$ is an $\eps$-repulsing supermartingale for 
 $\treach{C}$.
\end{lemma}

We now show how to obtain an upper bound on the probability 
of an invariant failure via repulsing supermartingales. Techniques used within 
the 
proof of Theorem~\ref{thm:old-ranking} (which are similar to the proof of Lemma 
5.5 
in~\cite{HolgerPOPL}) are not applicable, as they crucially rely on the fact 
that the supermartingale is non-negative before reaching $C$. Instead, we use a 
powerful tool of Martingale theory called Azuma's inequality.

\smallskip
\begin{theorem}[Azuma's inequality~\cite{Azuma67:ineq}]
Let $(\Omega,\mathcal{F},\probm)$ be a probability space and 
$\{X_i\}_{i=0}^{\infty}$ a supermartingale w.r.t. $\mathcal{F}$ with 
$c$-bounded differences. Then for each 
$n\in \Nset_0$ and each $\lambda >0$ it holds
\[
\probm(X_n - X_0 \geq \lambda) \leq e^{-\frac{\lambda^2}{2nc^2}}.
\]
\end{theorem}

Intuitively, Azuma's inequality provides exponentially decreasing tail bound on 
the probability that a supermartingale exhibits a large deviation from its 
expected value. In the following lemma (inspired by martingale use 
in~\cite{BBEK:oc-approx-IC}) the Azuma's inequality is used to obtain 
exponentially decreasing bound on 
probability that the set $C$ is reached in exactly $n$ steps.
\smallskip
\begin{lemma}
\label{lem:azuma-use}
Let $C$ be a set of configurations of an \APP{} $P$. Denote by $F_n$ the 
set of all runs $\run$ 
such that $\treach{C}(\run)=n$. Suppose that there 
exist $\eps>0$, $c>0$ and a linear $\eps$-repulsing supermartingale $\lem$ for 
$C$ supported by some pure invariant $\inv$ such that $\lem$ has 
$c$-bounded differences and $\lem(\locinit,\vecinit)< 0$. Then under each 
scheduler $\sigma$ it holds
\[
\probm^\sigma(F_n) \leq \alpha \cdot \gamma^n,
\]
where $\gamma = e^{-\frac{\eps^2}{2(c+\eps)^2}}$,
$\alpha= e^{\frac{\epsilon \cdot m_0}{(c+\eps)^2}}$ and 
$m_0=\lem(\locinit,\vecinit)$.
\end{lemma}
\begin{proof}[Key proof idea] 
We use $\lem$ to obtain a supermartingale 
$\{\tilde{X}_{i}\}_{i=0}^{\infty}$ with $c$-bounded differences such that for 
for each run $\run \in F_n$ it holds $\tilde{X}_n (\run) - \tilde{X}_0(\run) 
\geq n \cdot \eps -m_0$. We then apply the Azuma's inequality on 
$\{\tilde{X}_{i}\}_{i=0}^{\infty}$ to get the desired bound on the probability 
of $\tilde{X}_n (\run) - \tilde{X}_0(\run) 
\geq n \cdot \eps -m_0$ and thus also on the probability of $F_n$.
\end{proof}

\begin{proof}
Using Lemma~\ref{thm:rep-supermart-connect} we get from $\lem$ a stochastic 
process 
$\{X_{i}\}_{i=0}^{\infty}$ which is, for each scheduler $\sigma$, an 
$\eps$-repulsing
supermartingale 
 for the 
stopping time $\treach{C}$ with $c$-bounded differences. Now we define 
a stochastic 
process 
$\{\tilde{X}_i\}_{i=0}^{\infty}$ by putting 
\[
\tilde{X}_i(\run) = \begin{dcases}
X_i(\run) + i\cdot \eps & \text{if $\treach{C}(\run)\geq i$}\\
\tilde{X}_{i-1}(\run) & \text{otherwise}.
\end{dcases}
\]
Since $\{X_i\}_{i=0}^{\infty}$ is $\eps$-decreasing until $\treach{C}$, 
the process  
$\{\tilde{X}_i\}_{i=0}^{\infty}$ is a supermartingale. Moreover, it is easy to 
check that $\{\tilde{X}_i\}_{i=0}^{\infty}$ has $(c+\eps)$-bounded differences. 
Now for each $\run$ we have 
$\tilde{X}_0 (\run)=\lem(\locinit,\vecinit) < 0$. Moreover, from the 
definitions of 
$\{X_i\}_{i=0}^{\infty}$ and 
$\{\tilde{X}_i\}_{i=0}^{\infty}$ we get that $\run \in F_n$ implies 
$\tilde{X}_n(\run) = X_n(\run) + n\cdot \eps =\lem(\cfg{\sigma}{n}(\run)) + 
n\cdot \eps \geq n\cdot \eps$ (since $\lem$ assigns non-negative value to 
configurations in $C$ and $\run\in F_n$ is within $C$ in step $n$), and adding 
$-\tilde{X}_0(\run)$ to both sides yields $\run \in F_n \Rightarrow 
\tilde{X}_n(\run) 
 - \tilde{X}_0(\run)\geq n\cdot \eps -m_0$;
recall $m_0=\lem(\locinit,\vecinit)=\tilde{X}_0(\run)$. 
Hence, for each scheduler 
 $\sigma$ we have
\begin{equation}
\label{eq:azuma-prob-bound}
\probm^\sigma(F_n) \leq \probm^\sigma(\tilde{X}_n - \tilde{X}_0 \geq n\cdot 
\eps -m_0).
\end{equation}
Applying the Azuma's inequality for $\{\tilde{X}_i\}_{i=0}^{n}$ 
on~\eqref{eq:azuma-prob-bound} we get
\begin{align*}
\probm^\sigma(F_n) &\leq \probm(\tilde{X}_n - \tilde{X}_0 \geq n\cdot 
\eps-m_0)\\ &\leq  
\alpha \cdot e^{-\frac{n^2\cdot \eps^2}{2n(c+\eps)^2}} = \alpha \cdot \gamma^n,
\end{align*}
where $\alpha= e^{\frac{\epsilon \cdot m_0}{(c+\eps)^2}}$. 
\end{proof}

Using the above lemma we can bound the probability of reaching $C$ by a 
geometric series which can be easily evaluated. 

\smallskip
\begin{theorem}
\label{thm:repulsing-use-main}
Let $C$ be a set of configurations of an \APP{} $P$. Suppose that there 
exist $\eps>0$, $c>0$ and a linear $\eps$-repulsing supermartingale $\lem$ for 
$C$ supported by some pure invariant $\inv$ such that $\lem$ has 
$c$-bounded differences and $\lem(\locinit,\vecinit)< 0$. Then under each 
scheduler 
$\sigma$ it 
holds
\begin{equation}
\label{eq:quantitative-bound}
\probm^\sigma(\treach{C}<\infty)\leq 
\alpha \cdot 
\frac{\gamma^{\left\lceil{|\lem(\locinit,\vecinit)|}/{c}\right\rceil}}{1-\gamma},
\end{equation}
where $\gamma = e^{-\frac{\eps^2}{2(c+\eps)^2}}$ and 
$\alpha= e^{\frac{\epsilon \cdot m_0}{(c+\eps)^2}}$ and 
$m_0=\lem(\locinit,\vecinit)$ is the initial value.
\end{theorem}
\begin{proof}
For each $n$ let $F_n$ be as in Lemma~\ref{lem:azuma-use}. Denote by $A$ the 
number $\lceil|\lem(\locinit,\vecinit)|/c \rceil$. Observe that $F_n = 
\emptyset$ 
for each $n<A $. Indeed,  we need at least $A$ steps to reach $C$ 
from the initial configuration, because $\lem(\locinit,\vecinit)\leq 0$ (by the 
definition of a linear ranking supermartingale), the value of $\lem$ can 
increase by at most $c$ 
in each step, and reaching $C$ entails that the value of $\lem$ becomes 
non-negative. Hence, for each scheduler $\sigma$ we get
\begin{align*}
\probm^\sigma(\treach{C}<\infty)\! &=\! \sum_{n=0}^{\infty}\probm^\sigma(F_n)\! 
= \! 
\sum_{n=A}^{\infty}\probm^\sigma(F_n) \! \leq \! \sum_{n=A}^{\infty}\alpha \cdot \gamma^n 
\\ &=  \alpha \cdot \frac{\gamma^A}{1-\gamma},
\end{align*}
where the inequality in the middle comes from Lemma~\ref{lem:azuma-use}.
\end{proof}

\smallskip
\begin{example}[Illustration of Theorem~\ref{thm:repulsing-use-main}]\label{ex:repulthm}
Looking back at Example~\ref{ex:lrepsm},
the absolute value of the change in $\eta$ at each step is bounded from above 
by 12. 
Since the initial value of $\eta$ is $-3429$, we use Azuma's inequality and Theorem~\ref{thm:repulsing-use-main} 
to get the probability bound $5.06\cdot 10^{-6}$ on the violation of $\pinv$.
\end{example}

We now present the corollary that establishes the effectiveness of 
LRepSM for stochastic invariants.

\smallskip
\begin{corollary}
\label{coro:main}
Let $\pinv$ be a linear predicate map. Denote by $\neg\pinv$ the set of all 
configurations $(\loc,\vec{x})$ such that $\vec{x}\not\models\pinv(\loc)$. 
Assume that there exist $\eps>0$, $c>0$, and an $\eps$-LRepSM $\lem$ for 
$\neg\pinv$ with $c$-bounded differences such that $\lem(\locinit,\vecinit)<0$. 
Then $(\pinv,p)$ with $p = e^{\frac{\epsilon \cdot m_0}{(c+\eps)^2}}\cdot 
\frac{\gamma^{\left\lceil{|\lem(\locinit,\vecinit)|}/{c}\right\rceil}}{1-\gamma}$
 ($\gamma$ and $m_0$ are as in Theorem~\ref{thm:repulsing-use-main}) 
 is a stochastic invariant.
\end{corollary}

We note that the bound obtained from~\eqref{eq:quantitative-bound} is sound,
but not necessarily tight.
The magnitude of this bound crucially depends on $\lem$ and on initial valuation of variables. 
In Section~\ref{sec:exp} we discuss how to find a LRepSM $\lem$ providing good 
bounds in practice
(as also illustrated in Example~\ref{ex:repulthm}).

\section{Effectiveness of Repulsing Supermartingales}
In this section we discuss the effectiveness of repulsing supermartingales 
in several problems in analysis of probabilistic programs.

\subsection{Probabilistic Termination}
In Section~\ref{sec:repsm} we establish the effectiveness of repulsing 
supermartingales for stochastic invariants.
Theorem~\ref{thm:quantitative-termination} shows that stochastic 
invariants along with ranking supermartingales can obtain bounds for the 
probabilistic termination problem.
Hence the combination of repulsing and ranking supermartingales can 
answer the probabilistic termination problem.

\subsection{Refuting Almost-Sure and Finite Termination}

While a significant effort in analysis of non-probabilistic programs is devoted 
to proving termination, for bug-hunting purposes the analysis is often 
complemented by methods that aim to prove that a given program does \emph{not} 
terminate~\cite{GHMRX08:non-termination,VR08:nonterm-imperative, 
LNORR14:nonterm-smt,CCFNOH14:nonterm-safety,ABEL12:nonterm-multithreaded}. 
Similarly for probabilistic programs we can ask for refutation of almost-sure 
termination of a given program. We show how RepSMs can be used to this end.

If we have an $\eps$-LRepSM $\lem$ for the set of terminal configurations and 
the bound obtained from Theorem~\ref{thm:repulsing-use-main} is smaller than 
$1$, 
then $\lem$ in particular proves that the program does not terminate almost 
surely (from the given initial configuration). 
However repulsing supermartingales can refute a.s. termination  even for 
programs where the 
bound obtained by using Theorem~\ref{thm:repulsing-use-main} is $\geq 1$. 
To show this we use another 
powerful tool of martingale theory: the \emph{optional stopping theorem}.

\smallskip
\begin{theorem}[Optional Stopping, {\cite[Theorem 10.10]{Williams:book}}]
\label{thm:optional-stopping}
Let $(\Omega,\mathcal{F},\probm)$ be a probability space,
$\{X_i\}_{i=0}^{\infty}$ a supermartingale w.r.t. some filtration 
$\{\mathcal{F}_i\}_{i=0}^{\infty}$, 
and $\stime$ a 
stopping time w.r.t the same filtration. 
Assume that $\E[\stime]<\infty$ and $\{X_i\}_{i=0}^{\infty}$ has $c$-bounded 
differences for some $c\in \Rset$. 
Then
\[
\E[X_0] \geq \E[X_T].
\]
\end{theorem}

The optional stopping theorem guarantees that under given 
assumptions, the expected value of the supermartingale at the time of stopping 
(which can be, e.g. the time of program termination) is bounded from above by 
the expected initial value of the supermartingale. 
We can use the theorem to obtain the following.

\smallskip
\begin{theorem}
\label{thm:non-termination}
Let $C$ be a set of configurations of an \APP{} $P$. Suppose that there 
exist $\eps>0$, $c>0$ and a linear $\eps$-repulsing supermartingale $\lem$ for 
$C$ supported by some pure invariant $\inv$ such that $\lem$ has 
$c$-bounded differences. If $\lem(\locinit,\vecinit)<0,$ then under each 
scheduler $\sigma$ it 
holds
\begin{equation*}
\label{eq:nonterm}
\probm^{\sigma}(\treach{C}<\infty) <1.
\end{equation*}
\end{theorem}
\begin{proof}[Key proof idea]
Theorem~\ref{thm:repulsing-use-main} shows that the existence of $\lem$ implies 
the following: if a program execution reaches with positive probability a 
configuration $(\loc,\vec{x})$ such that 
$\lem(\loc,\vec{x})$ is below some sufficiently small negative number $A$ 
(whose 
magnitude depends only on $\lem$ and $c$), then the program does 
not terminate almost-surely. It thus suffices to prove that the program reaches 
such a configuration with positive probability under each scheduler $\sigma$. 
We define a stopping time $T$ that returns a first point in time in which 
we reach either $C$ or a configuration $(\loc,\vec{x})$ with 
$\lem(\loc,\vec{x})\leq 2A$ and apply the optional stopping theorem on the 
$\eps$-RepSM obtained from $\lem$. It can be proved that expectation of $T$ is 
finite, so optional stopping theorem applies to $T$. Now to get a contradiction 
we 
assume that a configuration with $\lem$-value smaller than $2A$ is reached with 
probability $0$. Then at time $T$ the current configuration is almost-surely in 
$C$ so the expected value of the supermartingale at time $T$ is non-negative. 
But the optional stopping theorem forces this expectation to be bounded from 
above by the initial value of the RepSM, i.e. by $\lem(\locinit,\vecinit)$, 
which is negative, a contradiction.
\end{proof}

Another important concept in the analysis of probabilistic programs is 
\emph{finite termination}~\cite{CFNH16:prob-termination}, sometimes also 
called~\emph{positive termination}~\cite{HolgerPOPL}. A program is said to 
terminate finitely if its 
expected termination time is finite. Of course, when a program terminates with 
probability less than 1 it is not finitely terminating. However, there are 
programs that terminate almost-surely but the expected termination time is 
infinite. Indeed, consider a program modelling a symmetric 1-dimensional random 
walk with a boundary:
\[\textbf{while } x\geq 0 \textbf{ do } 
\textbf{ if prob(} 0.5 \textbf{) then } x:=x+1 \textbf{ else } x:=x-1 
\textbf{ fi 
od}
\]
From the 
theory of random walks it follows that for each positive initial value of $x$ 
the program terminates almost-surely but its expected termination time is 
infinite.
Even for such programs the positive termination can be refuted, this time by 
using $0$-repulsing supermartingales.

\smallskip
\begin{theorem}
\label{thm:nonpositive}
Let $C$ be a set of configurations of an \APP{} $P$. Suppose that there 
exist $\eps\geq 0$, $c>0$ and a linear $\eps$-repulsing supermartingale $\lem$ 
for 
$C$ supported by some pure invariant $\inv$ such that $\lem$ has 
$c$-bounded differences. If $\lem(\locinit,\vecinit)<0$, then under each 
scheduler $\sigma$ it 
holds
\begin{equation}
\label{eq:nonpositive*}
\E^{\sigma}(\treach{C}) =\infty.
\end{equation}
\end{theorem}
\begin{proof}
Let $\{X_i\}_{i=0}^{\infty}$ be the $\eps$-repulsing supermartingale obtained 
from 
$\lem$ using Lemma~\ref{thm:rep-supermart-connect}. Assume, for the sake of 
contradiction, that there exists a scheduler $\sigma$ 
such that $\E^{\sigma}(\treach{C}) <\infty$. Since $\{X_i\}_{i=0}^{\infty}$ has 
$c$-bounded differences, using the optional stopping theorem we get
$\E^{\sigma}[X_0]\geq \E^{\sigma}[X_{\treach{C}}]$. But from the definition of 
$\lem$ we get $\E^{\sigma}[X_0]=\lem(\locinit,\vecinit)< 0$ and 
$\E^{\sigma}[X_{\treach{C}}]\geq 0 $, since $\lem$ attains non-negative values 
inside $C$. Hence, we derived a contradiction $0>0$.
\end{proof}

\begin{example}
For the 1-dimensional symmetric RW program pictured above 
(counter-example for finite-termination) it is easy to find a 
$0$-LRepSM $\lem$ for the set of terminal configurations 
supported by a pure invariant : say that $\lem$ is equal to $-x-1$ in all 
locations but the terminal one, where it is equal to $0$. The supporting 
invariant is, e.g. $x\geq-1$ for all locations. 
\end{example}

\subsection{Proving Almost-Sure Persistence}

The applicability of repulsing submartingales extends beyond reachability 
properties. In some applications of probabilistic programs, such as modelling 
of complex reactive systems~\cite{CVS16:martingale-recurrence-persistence}, it 
is customary to consider programs 
that are \emph{not} terminating but continue to execute forever, e.g. because 
they model a system which should continuously respond to inputs from the 
environment (e.g. a thermostat~\cite{AKLP10:hybrid-model-checking}). One of the 
basic properties of 
such programs is 
\emph{persistence}~\cite{CVS16:martingale-recurrence-persistence}. A set of 
configurations $C$ is said to be \emph{almost-surely persistent} if under each 
scheduler $\sigma$ it holds with probability $1$ that all but finitely many 
configurations along a run belong to $C$ (or in other words, that we will 
eventually see only configurations from $C$). 
In~\cite{CVS16:martingale-recurrence-persistence} they presented a method of 
proving almost-sure persistence via so called geometric supermartingales. We 
present an alternative proof technique based on combination of ranking and 
repulsing supermartingales.

\smallskip
\begin{theorem}
Let $C$ be be a set of configurations of some \APP{} $P$. Denote by $\neg C$ 
the set of all configurations of $C$ that do not belong to $C$. Assume that 
there 
exist the following:
\begin{compactenum}
\item 
An $\eps>0$, $c>0$, and an $\eps$-LRepSM $\lem$ with $c$-bounded differences 
for the 
set $\neg C $ supported by some pure invariant $\inv$.
\item
An $\eps>0$, $K<0$, and an $\eps$-LRSM for the set $D=\{(\loc,\vec{x})\mid 
\lem(\loc,\vec{x})\leq K \text{ and } \vec{x}\models\inv(\loc)\}$.
\end{compactenum}
Then the set $C$ is almost-surely persistent.
\end{theorem}
\begin{proof}[Key proof idea]
Item 2. ensures that from any reachable configuration we eventually reach the 
set $D$ with probability 1 (Theorem~\ref{thm:old-ranking}). Item 1. ensures, 
that each time we enter $D$ the probability that we never return 
back to $\neg C$ is positive (Theorem~\ref{thm:non-termination}). As a matter 
of fact, it 
can be shown that this probability is bounded away from zero by a number $p>0$ 
which depends only on $c$, $K$ and $\lem$, but not on a concrete configuration 
in which we enter $D$. Hence, the probability that we enter $D$ and after that 
reach $\neg C$ again at least $n$ times is at most $p^n$. For $n$ going to 
$\infty$ this converges to $0$, showing that the probability of infinitely 
often seeing a configuration from $\neg C$ is $0$.
\end{proof}

\begin{example}
As a simple example, consider the program $$\textbf{while } \mathit{true} 
\textbf{ 
do 
} x:=\textbf{sample}(\mathrm{Uniform}(-2,1)) \textbf{ od}.$$
For any $n\in \Zset$ let $C_n$ be the set of configurations in which the value 
of $x$ is at most $n$. For each such $n$ we have, inside the loop, a 
$\frac{1}{4}$-repulsing supermartingale $x-n$ for $\neg C_n$, and we also have 
a $\frac{1}{4}$-ranking supermartingale $x-n+1$ for the set 
$\{(\loc,\vec{x})\mid 
x-n\leq -1\}$ (both supported by invariants that are true 
everywhere). Since both supermartingales have bounded differences, we get that 
each set $C_n$ is persistent.
\end{example}

\section{Computational Results}

In this section we discuss computational aspects of our framework. Since 
synthesis of $\eps$-ranking supermartingales supported by a linear predicate 
map was already addressed in the previous 
work~\cite{SriramCAV,CFNH16:prob-termination}, we focus on algorithms related 
to those aspects of probabilistic reachability which are new, i.e. 
those related to stochastic invariants and repulsing supermartingales. Since 
our techniques are extensions of already known techniques for ranking 
supermartingales and invariant synthesis, we present only a high-level 
description.

We consider two main problems: 
\begin{compactenum}
\item 
For a given \APP{} $P$ with a given pure 
invariant $\inv$ and a linear predicate map $\pinv$ compute a number $p$ such 
that $(\pinv,p)$ is a stochastic invariant.
\item
For a given \APP{} $P$ and a set $C$ of configurations, compute a linear 
predicate map $\pinv$ and a number $p$ such 
that $(\pinv,p)$ is a stochastic invariant supporting
some $\eps$-LRSM for the set $C$.
\end{compactenum}

We assume that the set $C$ in item 2 above is given by some linear predicate 
map $\setmap$, i.e. it is a set of all $(\loc,\vec{x})$ such that 
$\vec{x}\models \setmap(\loc).$ This ensures that all the objects we work with 
are linear, which allows for a more efficient solution. In particular, the set 
of all terminal configurations can be easily given in this way, so point~2 
also 
concerns 
obtaining stochastic invariants for proving high-probability termination.

We start with presenting an algorithm for item 1 above, and then show how an 
algorithm for item 2 can be obtained as a straightforward generalization of 1.

We aim to compute the bound $p$ using Corollary~\ref{coro:main}, 
i.e. we want to compute an $\eps$-LRepSM for the set $\neg\pinv$ with 
$c$-bounded-differences supported by $\inv$. Note that $\neg \pinv$ can also be 
expressed by a linear predicate map, and this LPM can be computed in polynomial 
time provided that $\pinv$ is \emph{polyhedral}, i.e. that each $\pinv(\loc)$ 
is a linear assertion (a conjunction of linear inequalities). We call a set of 
configurations polyhedral if it can be defined by a polyhedral LPM.

We adapt a well known constrained-based method for generating linear 
ranking 
functions and (non-stochastic) invariants in non-probabilistic 
programs~\cite{CSS03:linear-invariant-farkas, 
DBLP:conf/tacas/ColonS01,DBLP:conf/vmcai/PodelskiR04}, which was adapted for 
synthesizing $\eps$-LRSMs in probabilistic 
programs~\cite{SriramCAV,CFNH16:prob-termination}. We briefly recall this 
approach and explain its adaptation. So suppose that we are given a program 
$P$, 
a polyhedral set $C$, 
and an LPM $\inv$, and we want to compute numbers $\eps>0$, $c>0$, and, 
for 
each location 
$\loc$ of $\pCFG_P$, coefficients
$b^{\loc},a_1^{\loc},\dots,a_{|\pvars|}^{\loc}$ such that the LEM $\lem$ given 
by $\lem(\loc)=b^\loc + \sum_{i=1}^{|\pvars|}a_i^{\loc}\cdot x_i$ is an 
$\eps$-LRSM for $C$ with $c$-bounded differences supported by $\inv$. Since 
LRSMs can be re-scaled 
by an arbitrary 
positive constant, we can assume that $\eps\geq 1$ and $c\geq 1$ (it always 
holds that $c\geq 
\eps$). We denote by $U$ the set $\{b^\loc,a_1^{\loc},\dots,a_{|\pvars|}^{\loc}\mid \loc 
\in \locs\}\cup\{c,\eps\}$.
The algorithm 
of~\cite{CFNH16:prob-termination} 
constructs a system of linear inequalities $\vec{y}\cdot{Z} \geq \vec{d}$ (here 
$Z$ is a matrix) that is \emph{adjusted to $P$, $C$, and $\inv$}, which means 
that each term in each inequality of the system has one of the following forms:
\begin{compactitem}
\item It is a variable with a name corresponding to some element of $U$.
\item It is a term of the form $y \cdot z$, where $y$ is a variable and $z$ is a coefficient (i.e. a number) appearing in $\inv$ (i.e. some inequality of $\inv$ contains a term of the form $z\cdot x_i$ for some $i$).
\item It is a term of the form $y \cdot z$, where $y$ is a variable and $z$ is 
a coefficient (i.e. a number) appearing in the LPM describing $C$.
\end{compactitem}
Moreover, any solution of the system yields an $\eps$-LRSM for $C$ with $c$-bounded 
differences (by substituting the solution values for variables in $U$). If the 
system is unsolvable, then no such LRSM exists.

Intuitively, the construction of $\vec{y}\cdot{Z} \geq \vec{d}$ proceeds as 
follows: the algorithm first translates the conditions in 
Definition~\ref{def:lrsm} into a conjunction of formulas of the form 
\begin{align}\exists 
u_1\dots\exists u_m \forall x_1 \dots\forall x_{|\pvars|} 
\,\varphi\Rightarrow \label{eq:formulae}
\psi,\end{align}
 where $u_1,\dots,u_m$ are all the elements of $U$, $\varphi$ is a 
linear 
assertion over variables  $\{x_1,\dots,x_{|\pvars|}\}$ whose coefficients are 
numbers that appear as coefficients in 
$\inv$ or in description of $C$, and $\psi$ is 
an arithmetic expression involving numbers and elements of $U\cup 
\{x_1,\dots,x_{|\pvars|}\}$ which is
linear 
if the elements of the set
$\{c,\eps,x_1,\dots,x_{|\pvars|}\}$ are taken as variables (in particular, its 
coefficients are 
independent of 
$\inv$ and $C$). The algorithm then utilizes Farkas's 
lemma~\cite{FarkasLemma} to convert each such formula into an equivalent 
existentially quantified linear assertion, i.e. a system of linear 
inequalities adjusted to $P$, $C$, and $\inv$. 

To obtain the required $\eps$, $c$ and LEM $\lem$ it thus 
suffices to solve the linear system $\vec{y}\cdot Z \geq \vec{d}$. The 
construction of the system can be done in polynomial time provided that $C$ is 
polyhedral. Note that in particular, the set of terminal 
configurations is polyhedral.

Now assume that instead of synthesizing an $\eps$-LRSM for some set $C$ we want 
to 
synthesize an $\eps$-LRepSM for $C=\neg\pinv$. The point is that $\neg\pinv$ is 
again expressed by a linear predicate map and all the formulas 
arising 
from conditions in the definition of an $\eps$-LRepSM again have the 
form~\eqref{eq:formulae}. This is easy to see as almost all conditions in the 
definition of a LRepSM are the same as for LRSM. The only difference is in the 
non-negativity condition, where in LRSMs we require non-negativity outside $C$, 
while in LRepSMs inside $C$. But for all locations $\loc$ both these 
constraints are of the form ``for 
all $\vec{x}$ satisfying a given linear predicate, 
$b^\loc + \sum_{i=1}^{|\pvars|}a_i^\loc \cdot x_i \geq 0$'', and thus can be 
transformed into a conjunction of formulae of the form~\eqref{eq:formulae}. 
Hence, we can again reduce computing an $\eps$-LRepSM for $\neg 
\pinv$ with $c$-bounded differences to solving a system of linear constraints, and the resulting system of linear constraints is again adjusted to $P$, $C$, and $\inv$.
The method is complete in the sense that an $\eps$-LRepSM for $C$ with 
$c$-bounded differences exists if and 
only if the system of linear inequalities $\vec{y}\cdot Z \geq \vec{d}$ has a 
solution. This 
is proved in the same way as for LRSMs in~\cite{CFNH16:prob-termination}.

As shown in Section~\ref{sec:repsm}, different LRepSMs can produce different 
upper bounds on the probability of reaching $\neg\pinv$. 
Theorem~\ref{thm:repulsing-use-main} shows that in order to get good bounds, it 
is 
vital that the computed LRepSM maximizes $|\lem(\locinit,\vecinit)|/c$. Since 
this 
function is not linear in $c$ and coefficients of $\lem$, we do not look for 
optimal $\lem$ and $c$ directly but instead we compute optimal LRepSMs $\lem$ 
for multiple heuristically chosen values of $c$ and then pick the one giving 
the best result.

The algorithm can be summarized as follows:
\begin{compactenum}
\item We fix $\eps=1$ (this is w.l.o.g. as LRepSMs can be rescaled 
arbitrarily).
\item Using the constrained based-approach described above, we compute the 
minimal $c$ such that there exists a $1$-LRepSM for $\neg \pinv$ with 
$c$-bounded differences. We do this by constructing the system of linear 
inequalities $\vec{y}\cdot Z \geq \vec{d}$ and use linear programming to 
minimize $c$ under the constraints given by the system. Denote by $c_{\min}$ 
the optimal $c$.
\item For some fixed number of iterations $N$ we do the following: for each 
$0\leq j \leq N$ we:
\begin{compactitem}
\item
Compute a $1$-LRepSM $\lem_j$ for $\neg\pinv$ such that $\lem_j$ has 
$(c_{\min}+j)$-bounded differences and minimizes $\lem_j(\locinit,\vecinit)$ 
(i.e. 
maximizes $|\lem_j(\locinit,\vecinit)|$). We do 
this again by constructing the system $\vec{y}\cdot Z \geq \vec{d}$ (we need to 
change just the terms referring to difference bound) and minimizing the 
objective 
function $\lem_j(\locinit,\vecinit)$ using LP subject to the constraints of the 
system (since $(\locinit,\vecinit)$ is given, the objective function is 
linear).
\item
Apply Theorem~\ref{thm:repulsing-use-main} on $\lem_j$ to get a bound $p_j$ on 
reaching $\neg\pinv$.
\end{compactitem}
\item We put $p=\min_{c_{\min}\leq j \leq c_{\min}+N} p_j$ and output 
$(\pinv,p)$ as a 
stochastic invariant.
\end{compactenum}

In our experiments we used $N=1000$.

Now we turn to problem 2, i.e. computation of a stochastic invariant 
$(\pinv,p)$ such that $\pinv$ supports a linear ranking supermartingale for 
some set 
$C$. Since $\pinv$ might have in principle unbounded size, we first have to fix 
a template for it, i.e. specify how many conjuncts and disjuncts can each 
$\pinv(\loc)$ consist of. This amounts to specifying a \emph{symbolic} linear predicate map $S\inv$, where coefficients in each linear inequality of $S\inv$ are not concrete numbers but abstract symbols. Note that symbolic LPMs can be also used to describe unknown sets of configurations.

 Now take a look back on the above algorithms for computing $\eps$-LRSM or 
 $\eps$-LRepSM for a given set $C$ supported by a given LPM $\inv$. Previously, 
 we used LPMs with concrete coefficients to encode both $C$ and $\inv$ on 
 input, but we can supplant these with symbolic LPMs, effectively parametrizing 
 the inputs $C$ and $\inv$. Since the original algorithms produce systems 
 adjusted to $P$, $C$ and $\inv$, 
 when the algorithms are run with a symbolic LPM instead of a concrete LPM on 
 input, they produce a system of \emph{quadratic} inequalities (as the 
 coefficients in $C$ and $\inv$ are now unknown). It can be easily shown that 
 there is a one-to-one correspondence between solutions of such a quadratic 
 system and tuples $(\lem,C,\inv,\eps,c)$, where $\lem$ is an $\eps$-LRSM (or 
 $\eps$-LRepSM, depending on which of the two algorithms we use) with 
 $c$-bounded differences for a set $C$ supported by an LPM $\inv$ such that $C$ 
 and $\inv$ can be formed by instantiating the unknown parameters with concrete 
 numbers. 
 
 We now construct two quadratic systems of inequalities: system 
 $\mathcal{S}_1$, produced by the LRSM algorithm on input $P$, $C$ (here $C$ is 
 a given set of configurations, i.e. a concrete set), and $S\inv$ (which is a 
 symbolic LPM encoding a template for the stochastic invariant we seek), and 
 system $\mathcal{S}_2$, produced by the LRepSM algorithm on input $P$, $\neg 
 S\inv$ (a set encoded by a symbolic LPM; a locationwise negation of the 
 aforementioned template), and $\mathit{True}$, where 
 $\mathit{True}$ is a (concrete) trivial invariant true in every location. Note 
 that in the first system we treat $S\inv$ as
a symbolic representation of an LPM, while in the second one we treat it as a symbolic representation of a set of configurations to 
avoid. We then identify the 
variables in~$S_1$ and $S_2$ referring to the same unknown coefficients in $S\inv$. Simultaneously solving both systems yields triples $(\lem,C,\pinv,\eps,c)$ and $(\lem',\neg\pinv,\mathit{True},\eps',c')$, where $\lem$ is an $\eps$-LRSM with $c$-bounded differences for $C$ supported by $\pinv$, and $\lem'$ is an $\eps'$-LRepSM with $c'$-bounded references for $\neg \pinv$ that can be used to bound the probability of violating $\pinv$.
We note that checking the solvability of a quadratic systems of 
inequalities can be done in PSPACE by reduction to existential first-order 
theory of reals. Also note that instead of $\mathit{True}$ we can use any other pure invariant.

\begin{theorem}
Existence of a LRepSM 
for a given set $C$ can be decided in polynomial time provided that $C$ is 
polyhedral. Existence of an LPM $\pinv$ such that $\pinv$ 
supports some 
LRSM for a given polyhedral set and at the same time $\neg\pinv$ admits an  
LRepSM can be 
reduced to existential first-order theory of reals and thus decided in PSPACE.
\end{theorem}

\section{Experimental Results}\label{sec:exp}

\begin{table*}
{\footnotesize
\begin{tabular}{|l|l|l|}
\hline
 & Initial Configuration & $\pinv$ Violation Probability Bound\\
\hline
\hspace{-1mm}Example 1\hspace{-3mm}& (i) $x:=10$, (ii) $x:=50$, (iii) $x:=100$ 
& (i) 
$5.1\cdot10^{-5}$, (ii) $1.0\cdot10^{-4}$, (iii) $2.5\cdot10^{-4}$\\
\hspace{-1mm}Example 2\hspace{-3mm}& (i) $x,y:=1000,10$, (ii) $x,y:=500,40$, 
(iii) $x,y:=400,50$ & (i) 
$2.4\cdot10^{-11}$, (ii) $5.5\cdot10^{-4}$, (iii) $1.9\cdot10^{-2}$ \\
\hspace{-1mm}Example 3\hspace{-3mm}& (i) $x,y,z:=100,100,100$, (ii) 
$x,y,z:=100,150,200$, (iii) $x,y,z:=300,100,150$\hspace{-1mm} & (i) 
$4.4\cdot10^{-17}$, (ii) 
$2.9\cdot10^{-9}$, (iii) $1.3\cdot10^{-7}$ \\

\hline
\end{tabular}
}
\caption{Experimental results}\label{tab:exp}
\end{table*}

In this section we present some basic experimental results for our methods. 
The experimental results are basic and to verify that the new concepts we introduce
are relevant. 
We consider three simple academic examples described below.
In the corresponding pseudocode, we present invariants in square brackets.
\begin{compactenum}
\item {\em Example 1:} The first example is a one-dimensional random walk which initially moves 
with higher probability to the left as compared to the right. However, if $x$ 
is incremented above $1000$, the process starts drifting away from zero, so the 
program does not terminate almost-surely.
The details of the example along with invariants is given in Figure~\ref{fig:a}.

\item {\em Example 2:} In the second example, we have two variables $x$ and 
$y$, and the program models a generalized 2-dimensional random walk. Variable 
$x$ tends to drift away from zero while $y$ tends to drift towards zero and 
thus towards satisfaction of the termination condition. However, if $x$ hits 
zero, the program gets stuck in an infinite loop, so we want to show that the 
probability of this happening is small.
The details of the example along with invariants is given in Figure~\ref{fig:b}.

\item {\em Example 3:} In the third example (Figure~\ref{fig:c}), we have three 
variables $x,y,z$.
For various combinations, with high probability we either decrease both $x$ and $y$,
or $z$, and with low probability we either increase both $x$ and $y$, or $z$.
But the increments and decrements are not proportional, and this is indeed a 3-dimensional
example. We note that the program in this example terminates almost-surely, 
but this does not simplify the computation of the probability bound for the 
given LPM $\pinv$.

\end{compactenum}
We consider various initial configurations of the examples. For each example we obtain a probability threshold
for a given LPM $\pinv$ (and thus obtain a stochastic invariant).
Our experimental results are shown in Table~\ref{tab:exp}.
In all the cases, our method creates a linear program which can be efficiently 
solved using 
any standard
solver (such as lpsolve~\cite{BEN:lpsolve}, CPlex~\cite{cplex}).

\lstset{language=affprob}
\lstset{tabsize=3}
\newsavebox{\expa}
\begin{lrbox}{\expa}
\begin{lstlisting}[mathescape]
while $x \geq 0$ do                 [$x \geq -2$]
	if $x \leq 1000$ then               [$x \geq 0$]
		if prob(0.5) then   [$x \geq 0$ and $x \leq 1000$]
			$x:=x-2$                 [$x \geq 0$ and $x \leq 998$]
		else
			$x:=x+1$                 [$x \geq 0$ and $x \leq 1001$]
		fi
	else
		if prob(0.5) then   [$x \geq 1001$]
			$x:=x-1$                 [$x \geq 1000$]
		else
			$x:=x+2$                 [$x \geq 1002$]
		fi
	fi
od

$\pinv:$ [$x \leq 1000$] at location 2
(first 'if'-branching),
'true' elsewhere.
\end{lstlisting}
\end{lrbox}

\lstset{language=affprob}
\lstset{tabsize=3}
\newsavebox{\expb}
\begin{lrbox}{\expb}
\begin{lstlisting}[mathescape]
while $1 \leq y$ do                 [$y \geq 0$]
	if prob(0.5) then      [$y \geq 1$]
		if prob(0.75) then  [$y \geq 1$]
			$x:=x+1$                 [$y \geq 1$]
		else
			$x:=x-1$                 [$y \geq 1$]
		fi
	else
		if prob(0.75) then  [$y \geq 1$]
			$y:=y-1$                 [$y \geq 0$]
		else
			$y:=y+1$                 [$y \geq 2$]
		fi
	fi                     [$y \geq 0$]
	while $x \leq 0$ do
		$x:=0$                    [$x \leq 0$ and $y \geq 0$]
	od
od

$\pinv$: [$x \geq 1$] at location 9
(entry of inner while-loop),
'true' elsewhere.
\end{lstlisting}
\end{lrbox}

\lstset{language=affprob}
\lstset{tabsize=3}
\newsavebox{\expc}
\begin{lrbox}{\expc}
\begin{lstlisting}[mathescape]
while $x \geq 0$ and $y \geq 0$ and $z \geq 0$ do 
          [$x \geq -1$ and $y \geq -1$ and $z \geq -1$]
   if prob(0.9) then                             
          [$x \geq 0$ and $y \geq 0$ and $z \geq 0$]
      if prob(0.5) then                           
          [$x \geq 0$ and $y \geq 0$ and $z \geq 0$]
         $x,y:=x-1,y-1$                           
          [$x \geq -1$ and $y \geq -1$ and $z \geq 0$]
      else $z:=z-1$                               
          [$x \geq 0$ and $y \geq 0$ and $z \geq -1$]
      fi
   else
      if prob(0.5) then                           
          [$x \geq 0$ and $y \geq 0$ and $z \geq 0$]
         $x,y:=x+0.1,y+0.2$                       
          [$x \geq 0$ and $y \geq 0$ and $z \geq 0$]
      else $z:=z+0.1$                             
          [$x \geq 0$ and $y \geq 0$ and $z \geq 0$]
      fi
   fi
od

$\pinv$: $[x+y+z \leq 1000]$ at location 1
(entry of outer loop),
'true' elsewhere.
\end{lstlisting}
\end{lrbox}

\begin{figure}[t]
\usebox{\expa}
\caption{Example~1.}\label{fig:a}
\end{figure}

\begin{figure}[t]
\usebox{\expb}
\caption{Example~2.}\label{fig:b}
\end{figure}

\begin{figure}[t]
\usebox{\expc}
\caption{Example~3.}\label{fig:c}
\end{figure}

\section{Related Work}

\smallskip\noindent{\em Probabilistic programs.} In the 70's and 80's, several 
semantic approaches for reasoning about probabilistic programs (including 
termination probabilities) were considered, most of them being based on 
probabilistic extensions of \emph{dynamic logic}~\cite{Pratt:dyn-logic}. 
In~\cite{Reif:prob-dl}, one such extension, PROB-DL, is applied to a restricted 
class of programs where there are no if-then-else branchings and no variable 
tests in loop guards (instead, loops are terminated according to a geometric 
distribution). A powerful probabilistic logic called \textbf{Pr}Dl was 
introduced in~\cite{FH:prdl}, allowing for first-order reasoning about events 
in the domain of computation and their effects on probabilities of assertions. 
The authors present an axiom system for \emph{Pr}Dl that is complete relatively 
to the underlying domain-specific logic (which might be undecidable in 
general), which allows one to check the validity of program properties 
"directly, (though,\dots, in general, not effectively)"~\cite{FH:prdl}. 
Decidable propositional fragments of probabilistic dynamic logic were studied 
in~\cite{Kozen:probabilistic-PDL,Feldman:propositional-probdl}, although as 
noted in~\cite{Kozen:probabilistic-PDL}, for practical verification purposes 
these would need to be extended with logic for reasoning about the computation 
domain. Moreover, none of the above approaches consider programs with 
non-determinism.

In the realm of probabilistic programs with non-determinism, the termination problems for probabilistic programs with discrete 
choices have been considered in~\cite{MM04,MM05}, but for probabilistic 
programs with infinite-state space and choices, only the qualitative 
problem has been studied.
The qualitative problem of almost-sure termination has been considered in 
several works such 
as~\cite{SriramCAV,BG05,HolgerPOPL,CFNH16:prob-termination,CFG16:positivstellensatz-arxiv}.
The termination for concurrent probabilistic programs under fairness was considered 
in~\cite{SPH84}.
A sound and complete characterization of almost-sure termination for countable 
state space was given in~\cite{HS85}.
A sound and complete method for proving termination of finite-state programs
was given in~\cite{EGK12}.
All previous works either consider discrete probabilistic choices or finite-state
space, or for general probabilistic programs consider the qualitative problem of 
almost-sure termination.
All works for almost-sure termination consider RSMs w.r.t. a given invariant.
In contrast, in this work we consider stochastic invariants and the probabilistic 
termination problem, and our results are applicable to probabilistic programs
with infinite-state space.

The use of martingales in probabilistic program analysis extends beyond 
termination properties. In~\cite{DP:book} martingales are used to derive bounds 
on expected termination time of randomized algorithms. 
In~\cite{CS14:expectation-invariants} they introduce \emph{expectation 
invariants} for single-loop probabilistic programs, which are statements about 
\emph{expected value} of program 
expressions whose validity is invariant during the program execution. In 
contrast, our stochastic invariant approach reasons about the 
\emph{probability} of a given assertion's validity.

There is also work on establishing a probability that a certain assertion 
holds. In~\cite{DBLP:conf/pldi/SampsonPMMGC14} 
they consider approximating the probability of assertions 
using optimized simulation, under semantics assuming terminating while loops. 
A method for approximating a probability of assertion based on symbolic 
execution was given in~\cite{Sankaranarayanan:2013:SAP:2491956.2462179}, where 
they 
also assume almost-surely terminating programs. Several works considered 
approximating the behaviour of probabilistic programs by abstracting them into 
finite Markov chains or MDPs~\cite{HWZ08,kattenbelt2009abstraction}. On the 
other hand, our repulsing 
supermartingales 
do not need any abstraction or simulation techniques to work, although we 
conjecture that they could be 
fruitfully 
combined with abstraction approaches to ``cut away'' configurations that are 
unlikely to be reached and thus reduce the size of the abstractions. 
In~\cite{DBLP:conf/sas/KatoenMMM10}, a Hoare-style calculus based on 
\emph{weakest pre-expectations} is used to reason about probabilistic effects 
of 
terminating programs, with a practical application presented 
in~\cite{GKM13:prinsys-expectation-invariant-tool}. The weakest pre-condition 
style of reasoning was also adapted for reasoning about expected running times 
of probabilistic programs~\cite{KKMO16:wp-expected-runtime}.

\smallskip\noindent{\em Non-probabilistic programs.}
Termination analysis of non-probabilistic programs has received a lot of
attention over the last decade, such as~\cite{DBLP:conf/cav/BradleyMS05,DBLP:conf/tacas/ColonS01,DBLP:conf/vmcai/PodelskiR04,DBLP:conf/pods/SohnG91,BMS05b,CSZ13,LJB01}.
Most of these works consider various notions of ranking functions for termination.
RSMs are a generalization of ranking functions, which has been studied for 
almost-sure
termination.
The extension of almost-sure termination to probabilistic termination needs
new conceptual ideas and methods which we present in this work.

\section{Conclusion and Future Work}
We considered the basic quantitative question of probabilistic
termination for probabilistic programs.
We introduced stochastic invariants for probabilistic termination, 
and repulsing supermartingales as the new concept that allows us to analyse
the problem of probabilistic termination.
There are several directions for future work.
The first one is to consider special cases of non-linear repulsing 
supermartingales
(such as polynomial repulsing supermartingales), and study whether efficient algorithmic 
approaches can be developed for them as well.
The second interesting direction is to consider practical approaches for the synthesis of 
stochastic invariants, as the theoretical results use the existential first 
order
theory of the reals.

\section*{Acknowledgements} This research was partially supported by Austrian 
Science Fund (FWF) NFN Grant No S11407-N23 (RiSE/SHiNE), ERC Start grant 
(279307: Graph Games), and Vienna Science and Technology Fund (WWTF) through 
project ICT15-003. The research leading to these results has received funding 
from the People Programme (Marie Curie Actions) of the European Union's Seventh 
Framework Programme (FP7/2007-2013) under REA grant agreement no [291734]. 
\DJ{}or\dj{}e \v{Z}ikeli\'c participated in this research during a research 
visit at IST Austria, funded by the OeAD Sonderstipendien, IST AUSTRIA 
programme awarded by the Austrian Agency for International Cooperation in 
Education and Research.

\bibliographystyle{abbrvnat}
\bibliography{bibliography-master,new}
\clearpage
\appendix

\begin{center}
{\Large Supplementary Material}
\end{center}

\section{Details of Section~\ref{sec:prelim}}

\smallskip\noindent{\em Notation for satisfaction of unary PLP $\phi$.}
For a unary PLP $\phi$, i.e. a PLP containing just a single variable $x$, we 
sometimes identify $\exsat{\phi}$ with the corresponding set of numbers defined by $\phi$, i.e. 
the set $\{\nu(x)\mid\nu \models \phi\}$. It will be always clear from the 
context whether $\exsat{\val}$ is deemed a set of valuations or of numbers.

\subsection{Details of Section~\ref{subsec:syntax}}
In this subsection we present the details of the syntax of affine probabilistic programs.

Recall that $\vars$ %
is a collection of 
\emph{variables}.
Moreover, let $\mathcal{D}$ be a set of \emph{probability distributions} on 
real numbers.
The abstract syntax of affine
probabilistic programs (\APP s)
is given by the grammar in Figure~\ref{fig:syntax}, where
the expressions $\langle \mathit{pvar}\rangle$ and $\langle
\mathit{dist}\rangle$  range over $\vars$ and $\mathcal{D}$, respectively.
We allow for non-deterministic assignments, expressed by a statement $x:=
\text{\textbf{ndet($\mathit{dom}$)}}$, where $\mathit{dom}$ is a \emph{domain
specifier} determining the set from which the value can be chosen (e.g. 
specifier $\text{\textbf{Int[$a,b$]}}$ restricts the choice to integers in the 
closed
interval $[a,b]$).  
The grammar is such that $\langle \mathit{expr} \rangle$ %
may evaluate to an arbitrary affine expression over the
program variables
Next, $\langle
\mathit{bexpr}\rangle$ may evaluate to an arbitrary propositionally linear
predicate.

The guard of each if-then-else statement is either $\star$, 
representing a (demonic) non-deterministic choice between the branches,
a keyword \textbf{prob}($p$), where $p\in [0,1]$ is a number given in decimal
representation (represents a probabilistic choice, where  the if-branch is
executed with probability $p$ and the then-branch with probability $1-p$), or
the guard is a propositionally linear predicate, in which case the statement
represents a standard deterministic conditional branching.

We assume that each \APP{} $\program$ is preceded by an initialization preamble 
in 
which 
each variable appearing in $\program$ is assigned some concrete number. 
Regarding distributions, for each $d\in \mathcal{D}$ we assume the
existence of a program primitive denoted by '\textbf{sample($d$)}' implementing 
sampling from $d$. In practice, the distributions appearing in a program would 
be those for which sampling is 
provided by suitable libraries (such as 
uniform distribution over some interval, Bernoulli, geometric, etc.), 
but we 
abstract away from these implementation details. For the purpose of our 
analysis, it is sufficient that for each distribution $d$ appearing in the 
program the following characteristics: expected value $\expv[d]$ of $d$ and a 
unary PLP $\varphi_d$ such that the \emph{support} of $d$ (i.e. the smallest 
closed set of real numbers whose complement has probability zero 
under $d$)\footnote{In particular, a support of a \emph{discrete} probability 
distribution $d$ is simply the at most countable set of all points on a real 
line that have positive probability under $d$. For continuous distributions, 
e.g. a normal distribution, uniform, etc., the support is typically either 
$\Rset$ or some closed real interval. } satisfies 
$\support(d)\subseteq\exsat{\varphi_d}$.

\begin{figure}
\begin{align*}
\langle \mathit{stmt}\rangle &::= 
\langle \mathit{assgn} \rangle \mid \text{'\textbf{skip}'} \mid 
\langle\mathit{stmt}\rangle \, \text{';'} \, \langle \mathit{stmt}\rangle \\
&\mid   \text{'\textbf{if}'} \,
\langle\mathit{ndbexpr}\rangle\,\text{'\textbf{then}'} \, \langle
\mathit{stmt}\rangle \, \text{'\textbf{else}'} \, \langle \mathit{stmt}\rangle
\,\text{'\textbf{fi}'}
\\
&\mid  \text{'\textbf{while}'}\, \langle\mathit{bexpr}\rangle \,
\text{'\textbf{do}'} \, \langle \mathit{stmt}\rangle \, \text{'\textbf{od}'}
\\
\langle \mathit{assgn} \rangle &::= 
\,\langle\mathit{pvar}\rangle
\,\text{'$:=$'}\, \langle\mathit{rexpr} \rangle \mid 
\langle\mathit{pvar}\rangle \,\text{'$:=$}\,
\text{\textbf{ndet($\langle\mathit{dom}\rangle$)}'}
\\
&\mid \langle\mathit{pvar}\rangle \,\text{'$:=$}\,
\text{\textbf{sample($\langle\mathit{dist}\rangle$)}'}
\\
\vspace{0.5\baselineskip}
\langle\mathit{expr} \rangle &::= \langle \mathit{constant} \rangle \mid
\langle\mathit{pvar}\rangle
\mid \langle \mathit{constant} \rangle \,\text{'$\cdot$'} \,
\langle\mathit{pvar}\rangle
\\
&\mid \langle\mathit{expr} \rangle\, \text{'$+$'} \,\langle\mathit{expr} \rangle
\mid \langle\mathit{expr} \rangle\, \text{'$-$'} \,\langle\mathit{expr} \rangle
\\
\vspace{0.5\baselineskip}
\langle \mathit{dom} \rangle &::= \text{'\textbf{Int}'} \mid
\text{'\textbf{Real}'} \mid
\text{'\textbf{Int}$[\langle\mathit{constant}\rangle,\langle\mathit{constant}\rangle]$'}
 \\ 
 &\mid
\text{'\textbf{Real}$[\langle\mathit{constant}\rangle,\langle\mathit{constant}\rangle]$'}
\mid \langle \mathit{dom} \rangle \text{'\textbf{or}'}\langle \mathit{dom}
\rangle
\\
\vspace{0.5\baselineskip}
\langle \mathit{bexpr}\rangle &::=  \langle \mathit{affexpr} \rangle \mid
\langle \mathit{affexpr} \rangle \, \text{'\textbf{or}'} \,
\langle\mathit{bexpr}\rangle
\vspace{0.5\baselineskip}
\\
\langle\mathit{affexpr} \rangle &::=  \langle\mathit{literal} \rangle\mid
\langle\mathit{literal} \rangle\, \text{'\textbf{and}'}
\,\langle\mathit{affexpr} \rangle
\\
\langle\mathit{literal} \rangle &::= \langle\mathit{expr} \rangle\,
\text{'$\leq$'} \,\langle\mathit{expr} \rangle \mid \langle\mathit{expr}
\rangle\, \text{'$\geq$'} \,\langle\mathit{expr} \rangle
\\
&\mid \neg \langle \mathit{literal} \rangle
\\
\langle\mathit{ndbexpr} \rangle &::= {\star}\mid
\text{'\textbf{prob($p$)}'} \mid \langle\mathit{bexpr} \rangle
\end{align*}
\caption{Syntax of affine probabilistic programs (\APP 's).}
\label{fig:syntax}
\end{figure}

\subsection{Details of Section~\ref{subsec:semantics}}

\begin{remark}[Use of random variable]\label{rem:randuse}
In the paper we sometimes work with random variables 
that are functions of the type $R\colon\Omega \rightarrow S$ for some finite 
set $S$. These can be captured by the above definition by identifying the 
elements of $S$ with distinct real numbers.\footnote{This is equivalent to 
saying that a function $R\colon \Omega\rightarrow S$, with $S$ finite, is a 
random variable if for each $s\in S$ the set $\{\omega\in \Omega\mid 
R(\omega)=s\}$ belongs to $\mathcal{F}$.} The exact choice of numbers is 
irrelevant in such a case, as we are not interested in, e.g. computing expected 
values of such random variables, or similar operations. 
\end{remark}

\paragraph*{From Programs to pCFGs}
To every affine probabilistic program $P$ we can assign a pCFG $\pCFG_P$ whose 
locations correspond to the values of the
program counter of $P$ and whose transition relation captures the behaviour of
$P$. To obtain $\pCFG_{P}$, we first rename the variables in $P$ to 
$x_1,\dots,x_n$, where $n$ is the number of distinct variables in the program. 
The 
initial assignment vector $\vec{x}_0$ being specified in the program's 
preamble. The
construction of $\pCFG_P$ can be described inductively.
For each program $P$ the pCFG $\pCFG_P$ contains two distinguished
locations, $\ell^{\lin}_{P}$ and $\ell^{\lout}_{P}$, the latter one being always
deterministic, that intuitively represent the state of the program counter
before and after executing $P$, respectively. In the following, we denote by 
$\id_1$ a function such that for each $\vec{x}$ we have 
$\id_{1}(\vec{x})=\vec{x}[1]$.
\begin{compactenum}
\item {\em Deterministic Assignments and Skips.}
For $P= {x_j}{:=}{E}$ where $x_j$ is a program variable and $E$ is an arithmetic
expression, or $P = \textbf{skip}$, the pCFG $\pCFG_P$ consists only of
locations $\ell^{\lin}_P$ and $\ell^{\lout}_P$ (both deterministic) and a
transition $(\ell^{\lin}_{P},j,E,\ell^{\lout}_P)$ or
$(\ell^{\lin}_{P},1,\id_1,\ell^{\lout}_P)$, respectively.
\item {\em Probabilistic and Non-Deterministic Assignemnts}
For $P= {x_j}{:=}{\textbf{sample($d$)}}$ where $x_j$ is a program variable and 
$d$ is a distribution, the pCFG $\pCFG_P$ consists locations $\ell^{\lin}_P$ 
and $\ell^{\lout}_P$ and a
transition $(\ell^{\lin}_{P},j,d,\ell^{\lout}_P)$. For $P= 
{x_j}{:=}{\textbf{ndet($\mathit{dom}$)}}$, the construction is similar, with 
the only transition being $(\ell^{\lin}_{P},j,D,\ell^{\lout}_P)$, where $D$ is 
the set specified by the domain specifier $\mathit{dom}$.

\item {\em Sequential Statements.}
For $P = Q_1;Q_2$ we take the pCFGs $\pCFG_{Q_1}$, $\pCFG_{Q_2}$ and
join them by identifying the location $\ell^{\lout}_{Q_1}$ with
$\ell^{\lin}_{Q_2}$, putting $\ell^{\lin}_{P}=\ell^{\lin}_{Q_1}$ and
$\ell^{\lout}_{P}=\ell^{\lout}_{Q_2}$.

\item {\em While Statements.}
For $P = \textbf{while $\phi$ do }Q \textbf{ od}$ we add a new deterministic
location $\ell^{\lin}_{P}$ which we identify with $\ell^{\lout}_{Q}$, a new
deterministic location $\ell^{\lout}_{P}$, and transitions
$\tau=(\ell^{\lin}_{P},1,\id_1,\ell^{\lin}_{Q})$,
$\tau'=(\ell^{\lin}_{P},1,\id_1,\ell^{\lout}_{P})$ such that $G(\tau)=\phi$ and
$G(\tau')=\neg\phi$.

\item {\em If Statements.}
Finally, for $P = \textbf{if $\mathit{ndb}$ then }Q_1 \textbf{ else } Q_2
\textbf{ fi}$ we add a new location $\ell^{\lin}_{P}$ together with two
transitions $\tau_1 = (\ell^{\lin}_{P},1,\id_1,\ell^{\lin}_{Q_1})$, $\tau_2 =
(\ell^{\lin}_{P},1,\id_1,\ell^{\lin}_{Q_2})$, and we identify the locations 
$\ell^{\lout}_{Q_1}$ and $\ell^{\lout}_{Q_1}$ with $\ell^{\lout}_{P}$. (If both
$Q_j$'s consist of a single statement, we also identify $\ell^\lin_{P}$ with 
$\ell^{\lin}_{Q_j}$'s.) In this
case the newly added location $\ell^\lin_{P}$ is non-deterministic if and only 
if
$ndb$ is the keyword '$\star$'. If
$\mathit{ndb}$ is of the form $\textbf{prob($p$)}$, the location $\ell^\lin_{P}$
is probabilistic with $\probdist_{\ell^\lin_{P}}(\tau_1)=p$ and
$\probdist_{\ell^\lin_{P}}(\tau_2)=1-p$. Otherwise (i.e. if $\mathit{ndb}$ is a
propositionally linear predicate), $\ell^\lin_{P}$ is a deterministic location
with $G(\tau_1)=\mathit{ndb}$ and $G(\tau_2)=\neg \mathit{ndb}$.
\end{compactenum}
Once the pCFG $\pCFG_P$ is constructed using the above rules, we put
$G(\tau)=\textit{true}$ for all transitions $\tau$ outgoing from deterministic
locations whose guard was not set in the process, and finally we add a self-loop
on the location $\ell^{\lout}_P$. This ensures that the assumptions in
Definition~\ref{def:stochgame} are satisfied.
Furthermore note that for pCFG obtained for a program $P$, since the only
branching is conditional branching, every location $\loc$ has at most two
successors $\loc_1,\loc_2$.

\section{Details of Section~\ref{sec:invm}}
In this section we present various illustrations of the definitions of Section~\ref{sec:invm}.

\smallskip
\begin{example}[Illustration of pure linear invariant]
Consider the simple \APP{} in Figure~\ref{fig:invariant-running} modelling an 
asymmetric 1-dimensional random walk. An LPM $\inv$ such that $\inv(\loc_0)$ is 
$x\geq -1$, $\inv(\loc_1)$ is $x\geq 0$, and $\inv(\loc_2)$ is $x\leq 0$ is a 
pure 
linear invariant.
\end{example}

\smallskip
\begin{example}[Illustration of supermartingale]
Consider again the program in Figure~\ref{fig:invariant-running} and define a 
stochastic process $\{X_{i}\}_{i=0}^{\infty}$ such that $X_i$ is the value of 
$x$ after $i$ steps of program execution. Then $\{X_{i}\}_{i=0}^{\infty}$ is a 
supermartingale w.r.t. canonical filtration: no matter the previous history of 
execution, if the current location is $\loc_0$ or $\loc_2$, the variable $x$ 
will not change in the next step (in particular, it will not increase). If the 
current location is $\loc_1$ and the current value of the variable is $n$, then 
in the next step the expected value of $x$ is 
$\frac{3}{4}(n-1)+\frac{1}{4}(n+1)=n-\frac{1}{2}\leq n$.
\end{example}

\smallskip
\begin{example}[Illustration of ranking supermartingales]
\label{ex:rsm}
Going back to the program in Figure~\ref{fig:invariant-running}: For $\eps>0$ 
the process 
$\{X_i\}_{i=0}^{\infty}$, where $X_i$ is the value of $x$ after $i$ steps, is 
\emph{not} an $\eps$-RSM for the termination time $\ttime$. 
This is because when we are in location $\loc_0$ and $x\geq 0$, then in the 
next step the value of $x$ does not change, in particular in does not decrease. 
However, the following process $\{X_i'\}_{i=0}^{\infty}$ \emph{is} an 
$\frac{1}{4}$-RSM for $\ttime$: if the location in step $i$ is $\loc_0$ or 
$\loc_1$, 
then for $i$ even $X_i'$ is the value of $x$ after $i$ 
steps, 
while for odd $i$, $X_i'$ is equal to the current value of $x$ minus 
$\frac{1}{4}$. If the location in step $i$ is $\loc_2$, then $X_{i+1}'$ equals 
$X_i$.
\end{example}

\smallskip
\begin{remark}
Definition~\ref{def:ranking} slightly diverges from the corresponding 
definition in, e.g.~\cite{CFNH16:prob-termination,HolgerPOPL}, where the RSM 
was required to decrease on average until it becomes non-positive, while we 
require it to decrease until the set $C$ is reached and reaching $C$ implies 
non-positivity of the RSM. While this is a rather technical difference (and 
indeed the earlier proofs related to RSM can be easily adapted to work with our 
definition), 
it allows us to easily apply RSM techniques on a larger class of  
programs.
\end{remark}

\smallskip
\begin{example}[Illustration of pre-expectation]
Again consider the program from Figure~\ref{fig:invariant-running}, and a LEM 
$\lem$ such that $\lem(\loc)=x$ for each location $\loc$. Then 
$\preexp{\lem}(\loc_1,5)=4*\frac{3}{4}+6\cdot\frac{1}{4}=\frac{9}{2}$ and $ 
\preexp{\lem}(\loc_0,7)=7$.
\end{example}

\smallskip
\begin{example}[Illustration of linear ranking supermartingales]
\label{ex:lrsm}
Again consider Figure~\ref{fig:invariant-running}, and take a LEM $\lem$ with 
$\lem(\loc_0)=x$, $\lem(\loc_1)=x-\frac{1}{4}$, and 
$\lem(\loc_2)=x-\frac{1}{4}$. Finally, consider an invariant $\inv$ such that 
$\inv(\loc_0)\equiv\inv(\loc_1) \equiv x\geq 0$ and 
$\inv(\loc_2)\equiv\mathit{true}$. Then $\lem$ is a $\frac{1}{4}$-LRSM for the 
set of terminal configurations
supported by $\inv$. However $\lem$ is \emph{not} a $\frac{1}{4}$-LRSM 
for the same set
if supported by an invariant $\inv'$ which assigns $\mathit{true}$ to every 
location.
This is because there are $x\in \Rset$ such that $\lem(\loc_0,x)<0$.
\end{example}

\section{Proofs}

\subsection{Proof of Lemma~\ref{thm:pure-supermart} and 
Lemma~\ref{thm:rep-supermart-connect}}

We present just the proof for the first Lemma, as for the second one the proof 
is almost identical.

\smallskip
\begin{reflemma}{thm:pure-supermart}
Let $P$ be an \APP{} and $\lem$ a linear $\eps$-ranking 
supermartingale for some set $\confset$ of  configurations of $\pCFG_P$ 
supported by $\inv$. 
Let $\neg \inv$ be the set of all configurations $(\loc,\vec{x})$ such that $\val_\vec{x}\not\models 
\inv(\loc)$.
Finally, let
$\{X_i\}_{i=0}^{\infty}$ be a stochastic process defined by 
\[
X_i(\run) = \begin{cases}
\lem( 
\cfg{\sigma}{i}(\run)) & \text{ if $\treach{C\cup \neg \inv} \geq i$} \\
X_{i-1}(\run) & \text{otherwise}.
\end{cases}
\]
Then under each scheduler $\sigma$ the stochastic process 
$\{X_i\}_{i=0}^{\infty}$ is an $\eps$-ranking 
supermartingale for $\treach{C\cup \neg \inv}$. 
Moreover, if $\lem$ has $c$-bounded differences, then so has 
$\{X_i\}_{i=0}^{\infty}$.
In  particular, if $\inv$ is a pure invariant of $P$, then 
 $\{X_i\}_{i=0}^{\infty}$ is an $\eps$-ranking 
supermartingale for $\treach{C}$.
\end{reflemma}
\begin{proof}
Denote by $\locs$ the set of program 
locations of $\pCFG_P$. 

Clearly, 
$\{X_i\}_{i=0}^{\infty}$ is adapted to the canonical filtration. Next, we need 
to prove that $\E^{\sigma}[|X_i|]<\infty$ for each $\sigma$ and $i$. Since each 
$X_i(\run)$ is an affine function of $\vec{x}^\sigma_i(\run)$, it suffices to 
prove that for each $i\in\Nset_0$ and each $1\leq j \leq |\pvars|$ it holds 
$\E^\sigma[|\vec{x}^\sigma_i[j]|]<\infty$, or equivalently 
$\E^\sigma[\vec{x}^\sigma_i[j]]\in \Rset$.
We proceed 
by induction on $i$. 
Since $\vec{x}^\sigma_i(\run)=\vecinit$ for each $\run$, the base 
case holds trivially. 
Now assume that $\E^{\sigma}[\vec{x}^\sigma_i[j]]\in \Rset$ for each applicable 
$j$. For 
all $k\geq i$, $1\leq j \leq |\pvars|$ we 
have
$\E^{\sigma}[\vec{x}_{k}^\sigma[j]] = 
\sum_{\loc\in\locs}\E^{\sigma}[\vec{x}_{k}^\sigma[j]\cdot\vec{1}_{\loc_i^\sigma
 = \loc}].
$
In particular, for each $\loc$ we have 
$\E^{\sigma}[|\vec{x}_{i}^\sigma[j]|\cdot\vec{1}_{\loc_i^\sigma
 = \loc}]<\infty$ and it suffices to prove that for each $\loc$ it holds 
 $\E^{\sigma}[|\vec{x}_{i+1}^\sigma[j]|\cdot\vec{1}_{\loc_i^\sigma
  = \loc}]<\infty$ for all $j$. So let $\loc$ be arbitrary. 
   Now if $\loc$ is probabilistic or 
  non-deterministic, then no transition outgoing from $\loc$ changes the values 
  of the variables, so for each $\run$ such 
    that $\loc^{\sigma}_{i}(\run)=\loc$ and each $1\leq j \leq |\pvars|$ 
  it holds 
  $\E^{\sigma}[\vec{x}_{i+1}^\sigma[j]\cdot\vec{1}_{\loc_i^{\sigma}=\loc}] =
  \E^\sigma[\vec{x}_{i}^\sigma[j]\cdot\vec{1}_{\loc_i^\sigma=\loc}])\in\Rset$.
  If $\loc$ is deterministic, then using similar reasoning as above it 
  suffices 
  to prove that for each $j$ and each transition $\tau$ outgoing from $\loc$ it 
  holds
  $ \E^{\sigma}[\vec{x}_{i+1}^\sigma[j]\cdot \vec{1}_{\loc_i^{\sigma}=\loc 
  \cap 
  \vec{x}_i^\sigma\models G(\tau)}]\in\Rset $. So fix some $j$ and $\tau$ and 
  let $k$ be the index of the variable modified by $\tau$. If $j\neq k$, then $ 
  \E^{\sigma}[\vec{x}_{i+1}^\sigma[j]\cdot \vec{1}_{\loc_i^{\sigma}=\loc \cap 
    \vec{x}_i^\sigma\models G(\tau)}]=\E^{\sigma}[\vec{x}_{i}^\sigma[j]\cdot 
    \vec{1}_{\loc_i^{\sigma}=\loc \cap 
        \vec{x}_i^\sigma\models G(\tau)}]\in\Rset$. If $j=k$, then we 
        distinguish three cases. If the update element of $\tau$ is a linear 
        function given by an expression $b+\sum_{\ell=1}^{|\pvars|}a_\ell\cdot 
        x_\ell$, then 
\begin{align}&\nonumber \E^{\sigma}[\vec{x}_{i+1}^\sigma[j]\cdot 
        \vec{1}_{\loc_i^{\sigma}=\loc \cap
            \vec{x}_i^\sigma\models G(\tau)}] \\ \nonumber
            =&\,\E^{\sigma}[(b+\sum_{\ell=1}^{|\pvars|}a_\ell\cdot 
                    \vec{x}_i^\sigma[\ell])\cdot 
            \vec{1}_{\loc_i^{\sigma}=\loc \cap 
                \vec{x}_i^\sigma\models G(\tau)}]  \\ \nonumber
=&\,\E^{\sigma}[(b\cdot 
\vec{1}_{\loc_i^{\sigma}=\loc \cap 
\vec{x}_i^\sigma\models G(\tau)}+\sum_{\ell=1}^{|\pvars|}a_\ell\cdot 
\vec{x}_i^\sigma[\ell]\cdot 
\vec{1}_{\loc_i^{\sigma}=\loc \cap 
\vec{x}_i^\sigma\models G(\tau)})]\\
=&\, \E^{\sigma}[(b\cdot 
\vec{1}_{\loc_i^{\sigma}=\loc \cap 
\vec{x}_i^\sigma\models G(\tau)}]+\sum_{\ell=1}^{|\pvars|}a_\ell\cdot 
\E^\sigma[\vec{x}_i^\sigma[\ell]\cdot 
\vec{1}_{\loc_i^{\sigma}=\loc \cap 
\vec{x}_i^\sigma\models G(\tau)})]         \nonumber        
\end{align}
The first expectation on the last line is clearly finite and the other 
$|\pvars|$ expectations are finite by induction hypothesis. Hence, also the 
expression in the first line represents a finite number.
If the update element of $\tau$ is an integrable distribution $d$, 
then $\E^{\sigma}[\vec{x}_{i+1}^\sigma[j]\cdot 
                    \vec{1}_{\loc_i^{\sigma}=\loc \cap
                        \vec{x}_i^\sigma\models 
                        G(\tau)}] = e\in\Rset$, where $e$ is the finite 
                        expectation of $d$. Similar 
                        argument covers the case when the assignment is 
                        non-deterministic, as $\sigma$ can use only integrable 
                        distributions.

It remains to prove that $\{X_i\}_{i=0}^{\infty}$ is an $\eps$-ranking 
(resp. 
$\eps$-repulsing) 
supermartingale for $\treach{C \cup \neg \inv}$. 
We show the 
argument for the ranking case, the repulsing case is similar. 
The implication 
$\treach{C\cup \neg\inv}(\run)>j \Rightarrow X_j(\run) \geq 0$ holds thanks to 
the first requirement in the definition of a linear ranking supermartingale. It 
thus suffices to prove that $\{X_i\}_{i=0}^{\infty}$ is $\eps$-decreasing until 
$\treach{C\cup \neg \inv}$. Define a random variable $Z$ such that
$
Z(\run) = 
\preexp{\lem}(\cfg{\sigma}{i}(\run))$  if $\loc^\sigma_i(\run)$ is stochastic 
or 
it is deterministic and the update in the transition enabled in 
$\cfg{\sigma}{i}$ is deterministic or stochastic; 
$Z(\run)=\sum_{\tau=(\loc,j,u,\loc')\in\Out(\loc^\sigma_i(\run))}{\sigmat(\run_{\leq
 i})(\tau)\cdot
 \lem(\loc',\vec{x}_i^\sigma(\run))}$ if $\loc_i^\sigma(\run)$ is 
 non-deterministic (here $\run_{\leq i}$ is the prefix of $\run$ of length 
 $i$); and $Z(\run)=\lem(\loc',\vec{x}_i^\sigma(j/a))$ if $\loc_i^\sigma(\run)$ 
 is deterministic $\vec{x}_i^\sigma(\run)\models G(\tau)$ s.t. 
 $\tau=(\loc,j,u,\loc')$, $u$ is a set, and the distribution 
 $\sigmaa(\run_{\leq i})$ has expected value $a$.
 
It is 
straightforward to check that $Z(\run)\leq 
\preexp{\lem}(\cfg{\sigma}{i}(\run))$ 
for all $\run$ 
and that for all 
sets $A\in \natfilt_i$ it holds 
$\E^\sigma[Z\cdot \vec{1}_A] = 
\E^\sigma[\lem(\cfg{\sigma}{i+1})\cdot\vec{1}_{A}]$. 
Hence, $\E^\sigma[\lem(\cfg{\sigma}{i+1})|\natfilt_i]\leq 
\preexp{\lem}(\cfg{\sigma}{i})$. From the definition of a linear ranking 
supermartingale it follows that for all $\run$ with $\treach{C\cup \neg 
\inv}(\run)>i$ we have $\preexp{\lem}(\cfg{\sigma}{i}(\run))\leq 
\lem(\cfg{\sigma}{i}(\run))-\eps$. 

But at the same time we 
can 
check that 
$\E^{\sigma}[X_{i+1}\cdot\vec{1}_A]=\E^{\sigma}[(Z\cdot 
\vec{1}_{\treach{C\cup\neg\inv}>i}+X_i\cdot\vec{1}_{\treach{C\cup\neg\inv}\leq 
i} )\cdot\vec{1}_{A}]$.
 Hence, 
\begin{align*}\E^{\sigma}[X_{i+1}|\natfilt_i] &= 
Z\cdot 
\vec{1}_{\treach{C\cup\neg\inv}>i}+X_i\cdot\vec{1}_{\treach{C\cup\neg\inv}\leq 
i} \\
 &\leq\preexp{\lem}(\cfg{\sigma}{i})\cdot 
 \vec{1}_{\treach{C\cup\neg\inv}>i}+X_i\cdot\vec{1}_{\treach{C\cup\neg\inv}\leq
  i} \\
&\leq (\lem(\cfg{\sigma}{i})-\eps)\cdot 
 \vec{1}_{\treach{C\cup\neg\inv}>i}+X_i\cdot\vec{1}_{\treach{C\cup\neg\inv}\leq
i} \\
&=\lem(\cfg{\sigma}{i})\cdot 
 \vec{1}_{\treach{C\cup\neg\inv}>i}+X_i\cdot\vec{1}_{\treach{C\cup\neg\inv}\leq
i} \\&-\eps\cdot 
 \vec{1}_{\treach{C\cup\neg\inv}>i}\\
 &=X_i\cdot 
  \vec{1}_{\treach{C\cup\neg\inv}>i}+X_i\cdot\vec{1}_{\treach{C\cup\neg\inv}\leq
 i} ]\\&-\eps\cdot
  \vec{1}_{\treach{C\cup\neg\inv}>i} = X_i -\eps\cdot
    \vec{1}_{\treach{C\cup\neg\inv}>i},
\end{align*}
as required. If $\inv$ is a pure invariant, then $\treach{\neg I} = \infty$ for 
all runs $\rho$, so $\treach{C\cup \neg I} = \treach{C}$.

The conservation of a $c$-boundedness property is straightforward.
The desired result follows.
\end{proof}

\subsection{Proof of Theorem~\ref{thm:quantitative-termination}}

We prove the missing part in the proof of 
Theorem~\ref{thm:quantitative-termination}, i.e. the following technical 
theorem, which forms a slight generalization of Theorem~\ref{thm:old-ranking}.

\smallskip
\begin{theorem}\label{thm:old-ranking2} %
Let $P$ be an \APP{}, $\sigma$ a scheduler, and 
$(\Omega,\natfilt,\probm^\sigma)$ the corresponding probability space. 
Further, 
let $C$ be a set of 
configurations of $\pCFG_P$ such 
that there exist an $\eps>0$ and an $\eps$-ranking supermartingale  
$\{X_i\}_{i=0}^{\infty}$
for $\treach{C}$. Then 
\begin{compactenum}
\item
$\probm^{\sigma}(\treach{C}<\infty)=1$, i.e. the set $C$ is reached 
almost-surely
\item
$\E^{\sigma}[\treach{C}]<\E^{\sigma}[X_0]/\eps$.
\end{compactenum}
\end{theorem}

There are three key differences between the formulation of 
Theorem~\ref{thm:old-ranking2} and the formulation of Theorem 1 in 
\cite{CFNH16:prob-termination}.
\begin{enumerate}
\item
In~\cite{CFNH16:prob-termination} there is a proof for the case when $C$ is the 
set of all terminal configurations. It is straightforward to adapt the proof 
for a general set of configurations.
\item In~\cite{CFNH16:prob-termination} there is a slightly different 
definition of an $\eps$-ranking supermartingale: it is required to decrease 
until its value becomes non-negative, while we require it to decrease until $C$ 
is reached. It is again straightforward to adapt the proof for our case.
\item
Bearing the previous two points in mind, in~\cite{,CFNH16:prob-termination} 
there is effectively a proof of Theorem~\ref{thm:old-ranking2} under an 
additional assumption that  
the 
$\eps$-ranking supermartingale $\{X_i\}_{i=0}^{\infty}$  for $\stime$ 
satisfies an additional condition that 
 there exists $K<0$ such that with probability $1$ it holds $X_i > 
K$. We show how to get rid of this assumption using the stopped supermartingale 
property.
\end{enumerate}

We first prove a variant of the stopped supermartingale property apt for our 
purposes.

\smallskip
\begin{lemma}
\label{lem:stopped-ranking}
Let $(\Omega,\mathcal{F},\probm)$ be a probability space, and let 
$\{X_i\}_{i=0}^{\infty}$ be an $\eps$-ranking supermartingale for some 
stopping 
time $\stime$. Consider the stochastic process $\{\hat{X}_i\}_{i=0}^{\infty}$ 
given by
\[
\hat{X}_i(\omega) = \begin{cases}
X_i(\omega) & \text{if $\stime(\omega)>i$} \\
-\eps &\text{otherwise}.
\end{cases}
\]
Then $\{\hat{X}_i\}_{i=0}^{\infty}$ is again an $\eps$-ranking supermartingale 
for $\stime$. Moreover, if $\{X_i\}_{i=0}^{\infty}$ is $c$-bounded for some 
$c$, then $\{\hat{X}_i\}_{i=0}^{\infty}$ is $(c+\eps)$-bounded.
\end{lemma}

\begin{proof}
We immediately have that $\{\hat{X}_i\}_{i=0}^{\infty}$ is adapted to the 
natural filtration. Next, since each $|X_i|$, $|\hat{X}_i|$ are non-negative 
random variables and $|\hat{X}_i(\omega)|\leq |X_i(\omega)| + \eps$ for all 
$\omega$, so it follows that $\E^{\sigma}[|\hat{X}_i|]\leq \E[|X_i|]+\eps 
<\infty + \eps < \infty$.
Finally, we need to show that for each $i\geq 0$ the following holds for 
$\omega$ with probability 1: $\E[\hat{X}_{i+1}|\mathcal{F}_i](\omega) \leq 
\hat{X}_i(\run) - \eps\cdot\vec{1}_{\stime > i}(\omega)$.

We start by proving that for each $i$ it holds with probability 1 that 
\begin{equation}\label{eq:cond-stopped} 
\E[\hat{X}_{i+1}|\mathcal{F}_i](\omega)= \begin{cases} 
\E[X_{i+1}|\mathcal{F}_i] & \text{if $\stime(\omega)>i$} \\ -\eps 
&\text{otherwise}.\end{cases} \end{equation}

Since we know that conditional expectation is almost surely uniquely 
defined~\cite{Billingsley:book}, it is sufficient to prove that the function 
defined 
via~\eqref{eq:cond-stopped} satisfies~\eqref{eq:cond-exp}. Let $A\in 
\mathcal{F}_i$ be arbitrary. Denote $B=\{\omega\in\Omega\mid 
\stime(\omega)>i\}$ and $C=\Omega\setminus B$. Plugging $\hat{X}_{i+1}$ and 
$\mathcal{F}_i$ to the left-hand side of~\eqref{eq:cond-exp} we get 
\begin{align*}
\mathit{LHS}&=\E[\hat{X}_i\cdot \vec{1}_A] \\&= \E[\hat{X}_i\cdot 
\vec{1}_{A\cap B}] + \E[\hat{X}_i\cdot \vec{1}_{A\cap C}]\\
&= \E[X_i\cdot \vec{1}_{A\cap B}] + \E[-\eps\cdot \vec{1}_{A\cap C}],
\end{align*}
where the second equality follows from $A=(A\cap B)\cup (A\cap C)$, from the 
fact that $B$ and $C$ are disjoint, and from linearity of expectation, and the 
third equality follows from the fact that $\hat{X}_i\cdot \vec{1}_{A\cap B} 
(\omega)= X_i$ for all $\omega\in A\cap B$ and $0$ for all other $\omega$'s, 
while $\hat{X}_i\cdot \vec{1}_{A\cap C} (\omega)= -\eps$ for all $\omega\in 
A\cap C$ and $0$ for other $\omega$'s. Now let us plug the function $Z$ 
defined 
by~\eqref{eq:cond-stopped} to the right-hand side of~\eqref{eq:cond-exp}. We 
get 
\begin{align*}
\mathit{RHS}&=\E[Z\cdot \vec{1}_A] \\&= \E[Z\cdot \vec{1}_{A\cap B}] + 
\E[Z\cdot \vec{1}_{A\cap C}]\\
&= \E[\E[X_i|\mathcal{F}_i]\cdot \vec{1}_{A\cap B}] + \E[-\eps\cdot 
\vec{1}_{A\cap C}]\\
&= \E[X_i\cdot \vec{1}_{A\cap B}] + \E[-\eps\cdot \vec{1}_{A\cap C}],
\end{align*}
where the second equality holds by same reasoning as above, the third equality 
holds as by definition $Z(\omega)$ is either $\E[X_i|\mathcal{F}_i](\omega)$ 
or 
$-\eps$ depending on whether $\omega \in B$ or $\omega \in C$, and the last 
equality follows by definition of conditional expectation of $X_i$. 
Hence~$\mathit{LHS}=\mathit{RHS}$ and~$Z$ is indeed a conditional expectation 
of $\hat{X}_i$ given $\mathcal{F}_i$.

Now we prove that $\{\hat{X}_i\}_{i=0}^{\infty}$ 
satisfies~\eqref{eq:ranking-sup}. Let $\omega$ and $i$ be arbitrary. If 
$\stime(\omega)>i+1$ then $\hat{X}_i(\omega)=X_i(\omega)$, 
$\E[\hat{X}_{i+1}|\mathcal{F}_i](\omega)=\E[X_{i+1}|\mathcal{F}_i](\omega)$ 
(as 
proven above), and since $\{X_i\}_{i=0}^{\infty}$ is an $\eps$-ranking 
supermartingale for $\stime$, it holds 
$\E[\hat{X}_{i+1}|\mathcal{F}_i](\omega) 
\leq \hat{X}_i - \eps\cdot\vec{1}_{\stime > i}$. If $\stime(\omega)\leq i$, 
then $\E[\hat{X}_{i+1}|\mathcal{F}_i](\omega) = -\eps $ and $ 
\hat{X}_i(\omega) 
- \eps\cdot\vec{1}_{\stime > i} = - \eps - \eps\cdot 0 = -\eps$, 
so~\eqref{eq:ranking-sup} holds. If $\stime(\omega)=i+1$, then 
$\E[\hat{X}_{i+1}|\mathcal{F}_i](\omega)=-\eps$, as proved above, and 
$\hat{X}_i(\omega) - \eps\cdot\vec{1}_{\stime > i}= X_i(\omega) - \eps\cdot 1  
\geq -\eps$, since $X_i(\omega)\geq 0$ as $\{X_i\}_{i=0}^{\infty}$ is 
$\eps$-ranking and $\stime(\omega)>i$. Hence~\eqref{eq:ranking-sup} holds also 
in this case.

The boundedness property is straightforward. The desired result follows.
\end{proof}

Now we prove Theorem~\ref{thm:old-ranking2}.

\smallskip
\begin{reftheorem}{thm:old-ranking2} %
Let $P$ be an \APP{}, $\sigma$ a scheduler, and 
$(\Omega,\natfilt,\probm^\sigma)$ the corresponding probability space. 
Further, 
let $C$ be a set of 
configurations of $\pCFG_P$ such 
that there exist an $\eps>0$ and an $\eps$-ranking supermartingale  
$\{X_i\}_{i=0}^{\infty}$
for $\treach{C}$. Then 
\begin{compactenum}
\item
$\probm^{\sigma}(\treach{C}<\infty)=1$, i.e. the set $C$ is reached 
almost-surely
\item
$\E^{\sigma}[\treach{C}]<\E^{\sigma}[X_0]/\eps$.
\end{compactenum}
\end{reftheorem}
\begin{proof}  %
In~\cite[Theorem 1]{,CFNH16:prob-termination} it was proved that if 
there exists an 
$\eps$-ranking supermartingale $\{X_i\}_{i=0}^{\infty}$  for $\stime$ 
satisfying an additional condition that 
 there exists $K<0$ such that with probability $1$ it holds $X_i > 
K$ for all $i$, then the following holds: 
$\probm^{\sigma}(\treach{C}<\infty)=1$, i.e. the set $C$ is reached 
almost-surely, and $\E^{\sigma}[\treach{C}]<(\E^{\sigma}[X_0]-K)/\eps$. To 
obtain first two items of Theorem~\ref{thm:old-ranking} we apply these results 
on the $\eps$-ranking supermartingale~$\{\hat{X}_i\}_{i=0}^{\infty}$ defined 
as 
in Lemma~\ref{lem:stopped-ranking}. 
\end{proof}

\subsection{Proof of Theorem~\ref{thm:non-termination}}

\begin{reftheorem}{thm:non-termination}
Let $C$ be a set of configurations of an \APP{} $P$. Suppose that there 
exist $\eps>0$, $c>0$ and a linear $\eps$-repulsing supermartingale $\lem$ for 
$C$ supported by some pure invariant $\inv$ such that $\lem$ has 
$c$-bounded differences. If $\lem(\locinit,\vecinit)<0,$ then under each 
scheduler $\sigma$ it 
holds
\begin{equation*}
\label{eq:nonterm2}
\probm^{\sigma}(\treach{C}<\infty) <1.
\end{equation*}
\end{reftheorem}
\begin{proof}
Let $\gamma=e^{-\frac{\eps^2}{2(c+\eps)^2}}$ and 
$A(\loc,\vec{x})=\lceil|\lem(\loc,\vec{x})|/c\rceil$. If it holds 
$\gamma^{A(\locinit,\vecinit)}/(1-\gamma)<1$, 
then~\eqref{eq:nonterm2} follows directly from~\eqref{eq:quantitative-bound} 
(in 
Theorem~\ref{thm:repulsing-use-main}). Otherwise note that there is positive 
$z\in 
\Rset$ such that $\gamma^{z}/(1-\gamma)<1$. For each configuration 
$(\loc,\vec{x})$ with $\lem(\loc,\vec{x})\leq 0$ and $A(\loc,\vec{x})\geq z$ 
the probability of reaching 
$Z$ 
when starting in $(\loc,\vec{x})$ is smaller than 1 under any scheduler (by 
applying the above reasoning to a program obtained from $P$ by changing the 
initial configuration to $(\loc,\vec{x})$). It thus suffices to prove that 
under any scheduler $\sigma$ a configuration $(\loc,\vec{x})$ with 
$\lem(\loc,\vec{x})\leq -z\cdot c$ 
is 
reached 
with positive probability from $(\locinit,\vecinit)$.

Let 
$\{X_i\}_{i=0}^{\infty}$ be the $\eps$-repulsing supermartingale obtained from 
$\lem$ using Lemma~\ref{thm:rep-supermart-connect} and let $\{\overline{{X}_i} 
\}_{i=0}^{\infty}$ be the translation of $\{X_i\}_{i=0}^{\infty}$ by $2zc$: we 
put $\overline{X}_i = X_i + 2zc$ for each $i\in\Nset_0$. It is straightforward 
to 
check that under each scheduler $\sigma$ the process $\{\overline{{X}_i} 
\}_{i=0}^{\infty}$ is an $\eps$-ranking supermartingale (note that this is 
ranking 
supermartingale and not repulsing supermartingale) for stopping time 
$\treach{D}$, where $D = C \cup \{(\loc,\vec{x})\mid \lem(\loc,\vec{x})\leq 
-2 z c\}.$ From Theorem~\ref{thm:old-ranking} it follows that 
$\E^{\sigma}[\treach{D}]<\infty$. Moreover, $\{\overline{{X}_i} 
\}_{i=0}^{\infty}$ still has $c$-bounded differences. Thus, we can apply the 
optional stopping theorem to get
\begin{equation}
\label{eq:nonterm-ost-use}
2zc > \E^\sigma[\overline{X}_0] \geq \E^\sigma[\overline{X}_{\treach{D}}].
\end{equation}
Now for each run $\run$ such that $\treach{D}(\run)<\infty$ (as shown above, 
the probability of such runs is $1$) we have either 
$-c\geq \overline{X}_{\treach{D}}(\run)\leq 0$ 
(in case that $\lem(\cfg{\sigma}{\treach{D}}(\run))\leq -2zc$) or 
$\overline{X}_{\treach{C}}\geq 2zc$ (when $\cfg{\sigma}{\treach{D}}\in C$). 
It follows that $\E^\sigma[\overline{X}_0] 
\geq\E^\sigma[\overline{X}_{\treach{D}}] \geq -p\cdot c + 
(1-p)\cdot 2zc$, where $p$ is the probability that $\cfg{\sigma}{\treach{D}}\in 
C$. Hence, $p\geq \frac{2zc-\E^\sigma[\overline{X}_0]}{c+2zc}>0 $.

Note that the probability of reaching configuration with ``sufficiently high'' 
$\lem$-value is not only positive, but actually bounded away from zero by a 
number which depends only on $c$, $z$ and $\lem(\locinit,\vecinit)$. This will 
be useful later.
\end{proof}

\bibliographystyleadd{abbrvnat}
\bibliographyadd{diss}

\end{document}